\documentclass[a4paper,onecolumn,11pt,accepted=2024-09-19]{quantumarticle}
\pdfoutput=1
\usepackage[numbers,sort&compress]{natbib}
\usepackage{authblk}
\usepackage{tabu}
\usepackage{array}
\usepackage{amsmath}
\usepackage{amsthm}
\usepackage{amssymb}
\usepackage{algorithm}
\usepackage{braket}
\usepackage{subfig}
\usepackage{color}
\usepackage[dvipsnames]{xcolor}
\usepackage[english]{babel}
\usepackage{graphicx}
\usepackage{grffile}
\usepackage{wrapfig,epsfig}
\usepackage{epstopdf}
\usepackage{url}
\usepackage{color}
\usepackage{qcircuit}
\usepackage{epstopdf}
\usepackage{algpseudocode}
\usepackage[T1]{fontenc}
\usepackage{bbm}
\usepackage{comment}
\usepackage{dsfont}
\usepackage{mathtools}
\usepackage{enumitem}
\usepackage{pifont}
\newcommand{\cmark}{\textcolor{green}{\ding{51}}}%
\newcommand{\xmark}{\textcolor{red}{\ding{55}}}%
\usepackage{tikz}
\usetikzlibrary{arrows}
\usepackage[margin=1in]{geometry}
\graphicspath{{./figs/}}
\usepackage{hyperref}
\usepackage[nameinlink,capitalize]{cleveref}
\usepackage{textcomp}
\usepackage{multirow}
\newtheorem{theorem}{Theorem}[section]
\newtheorem{lemma}[theorem]{Lemma}

\newcommand{\vtheta}{\boldsymbol \theta}
\newcommand{\wh}{\widehat}
\newcommand{\rd}{\mathrm{d}}
\newcommand{\Or}{\mathcal{O}}

\newcommand{\wt}{\widetilde}

\newcommand{\Deltam}{\Delta_{\text{dom}}}
\newcommand{\ptail}{p_{\mathrm{tail}}}

\newcommand{\R}{\mathbb{R}}

\newcommand{\la}{\lambda}

\renewcommand{\d}{\mathrm{d}}

\renewcommand{\tilde}{\wt}
\renewcommand{\hat}{\wh}

\renewcommand{\d}{\mathrm{d}}

\DeclareMathOperator{\poly}{poly}

\DeclareMathOperator*{\argmin}{arg\,min}

\newcommand{\ZZ}{\mathbb{Z}}

\definecolor{mygreen}{RGB}{80,180,0}
\definecolor{b2}{RGB}{51,153,255}

\newcommand{\nc}{\newcommand}

\nc{\nnl}{\nn \\ &}  
\nc{\fot}{\frac{1}{2}} 
\nc{\oo}[1]{\frac{1}{#1}} 
\newcommand{\ben}{\begin{enumerate}}
\newcommand{\een}{\end{enumerate}}
\nc{\mc}{\mathcal}
\graphicspath{{figures/}}
\nc{\onenorm}[1]{\L\| #1 \R\|_1} 

\DeclareMathOperator*{\argmax}{arg\,max}

\nc{\Ra}{\Rightarrow}
\nc{\zo}{\{0,1\}}

\usepackage{authblk}
\makeatletter
\newenvironment{breakablealgorithm}
  {
   \begin{center}
     \refstepcounter{algorithm}
     \hrule height.8pt depth0pt \kern2pt
     \renewcommand{\caption}[2][\relax]{
       {\raggedright\textbf{\fname@algorithm~\thealgorithm} ##2\par}%
       \ifx\relax##1\relax 
         \addcontentsline{loa}{algorithm}{\protect\numberline{\thealgorithm}##2}%
       \else 
         \addcontentsline{loa}{algorithm}{\protect\numberline{\thealgorithm}##1}%
       \fi
       \kern2pt\hrule\kern2pt
     }
  }{
     \kern2pt\hrule\relax
   \end{center}
  }
\makeatother

\begin{document}

\title{Quantum Multiple Eigenvalue Gaussian filtered Search: an efficient and versatile  quantum phase estimation method}

\author{Zhiyan Ding}
\affiliation{Department of Mathematics, University of California, Berkeley, CA 94720, USA}
\email{zding.m@berkeley.edu}
\author{Haoya Li}
\affiliation{Department of Mathematics, Stanford University, CA 94305, USA}
\email{lihaoya@stanford.edu}
\author{Lin Lin}
\affiliation{Department of Mathematics, University of California, Berkeley, CA 94720, USA}
\affiliation{Applied Mathematics and Computational Research Division, Lawrence Berkeley National Laboratory, Berkeley, CA 94720, USA}
\affiliation{Challenge Institute for Quantum Computation, University of California, Berkeley, CA 94720, USA}
\email{linlin@berkeley.edu}
\author{Hongkang Ni}
\affiliation{Institute for Computational and Mathematical Engineering, Stanford University, CA 94305, USA}
\email{hongkang@stanford.edu}
\author{Lexing Ying}
\affiliation{Department of Mathematics, Stanford University, CA 94305, USA}
\affiliation{Institute for Computational and Mathematical Engineering, Stanford University, CA 94305, USA}
\email{lexing@stanford.edu}
\author{Ruizhe Zhang}
\affiliation{Simons Institute for the Theory of Computing, University of California, Berkeley, CA 94720, USA}
\email{rzzhang@berkeley.edu}
\maketitle

\begin{abstract}
  Quantum phase estimation is one of the most powerful quantum primitives. This work proposes a new approach for the problem of multiple eigenvalue estimation: Quantum Multiple Eigenvalue Gaussian filtered Search (QMEGS). QMEGS leverages the Hadamard test circuit structure and only requires simple classical postprocessing. QMEGS is the first algorithm to simultaneously satisfy the following two properties: (1) It can achieve the Heisenberg-limited scaling without relying on any spectral gap assumption. (2) With a positive energy gap and additional assumptions on the initial state, QMEGS can estimate all dominant eigenvalues to $\epsilon$ accuracy utilizing a significantly reduced circuit depth compared to the standard quantum phase estimation algorithm. In the most favorable scenario, the maximal runtime can be reduced to as low as  $\log(1/\epsilon)$. This implies that QMEGS serves as an efficient and versatile approach, achieving the best-known results for both gapped and gapless systems. Numerical results validate the efficiency of our proposed algorithm in various regimes.
\end{abstract}

\section{Introduction}\label{sec:intro}

Phase estimation is among the most powerful quantum primitives, offering eigenvalue estimates of a Hamiltonian $H$ when given quantum access to Hamiltonian simulation $\exp(-iHt)$ or block encoding of $H$. While initial algorithms, such as the textbook version of phase estimation~\cite{NielsenChuang2000}, depend on techniques such as the Quantum Fourier Transform (QFT) and multiple ancilla qubits, recent advances demonstrate that comparable results can be achieved with as few as one ancilla qubit. Despite the much simpler quantum circuit, these ``modern'' approaches often match or exceed the performance of their predecessors due to enhanced postprocessing techniques.  This improvement has been demonstrated in applications such as estimating the ground-state energy, which involves determining the smallest eigenvalue of $H$ and is an important application of phase estimation \cite{AbramsLloyd1999, GeTuraCirac2019, LinTong2020a, LinTong2022, dong2022ground, WanBertaCampbell2022, zwp22, wfz22, DingLin2023, ni2023lowdepth, PRXQuantum.4.040341}. More recently, this enhancement has also been evident in the broader and more general problem of Multiple Eigenvalue Estimation (MEE) \cite{Somma2019, Stroeks_2022, Cortes2021QuantumKS, doi:10.1137/21M145954X,  Dutkiewicz2022heisenberglimited, Ding2023simultaneous,ni2023lowdepth_2,shen2023estimating}.
A typical example of MEE is estimating the low-lying energies of Hamiltonian $H$, which has many applications, such as determining the electronic and optical properties of materials.

In order to solve MEE for a given quantum Hamiltonian $H\in\mathbb{C}^{M\times M}$, we assume the availability of an initial state $\ket{\psi}$ that contains several dominant modes. Specifically, let $\left\{(\lambda_m,\ket{\psi_m})\right\}^M_{m=1}$ represent pairs of eigenvalues and eigenvectors of $H$. We define $p_m=\left|\braket{\psi_m|\psi}\right|^2$ as the overlap between the initial state and the $m$-th eigenvector. Our primary assumption in this paper is the \emph{Sufficiently Dominant Condition}: there exists a set of indices $\mathcal{D}\subset {1,2,\cdots,M}$ such that $p_{\min}=:\min_{i\in\mathcal{D}}p_i>\ptail=:\sum_{i\in\mathcal{D}^c}p_i$, where $\mathcal{D}^c=\{1,2,\cdots,M\}\setminus\mathcal{D}$. The eigenvalues $\{\lambda_m\}_{m\in \mc{D}}$ are then called the dominant eigenvalues of $H$ with respect to the initial state $\ket{\psi}$ (or simply the dominant eigenvalues), and the associated eigenvectors are referred to as the dominant eigenvectors. For simplicity, we assume $\|H\|\leq \pi$, which also implies $\{\lambda_m\}_{m\in\mathcal{D}}\subset[-\pi,\pi]$. Our objective is to estimate the dominant eigenvalues $\left\{\lambda_m\right\}_{m\in\mathcal{D}}$. From a signal processing perspective, the \emph{Sufficiently Dominant Condition} allows us to differentiate signal and noise, providing a natural basis for analysis. Different versions of the condition have appeared in previous works~\cite{ni2023lowdepth_2,ni2023lowdepth,Ding2023simultaneous,DingLin2023,Stroeks_2022}. A notable instance is in ground-state energy estimation, where the condition is equivalent to the initial overlap $p_1$ between the initial state and the ground state being greater than 0.5. The $\mathcal{O}(1)$ error allowed by this condition is the key to the robustness of many algorithms based on phase estimation.

To estimate the dominant eigenvalues, we assume an oracle access to the Hamiltonian simulation $\exp(-itH)$ for any $t\in\mathbb{R}$. Specifically, given any $t\in\mathbb{R}$, we assume the ability to implement the Hadamard test circuit (see \cref{sec:idea} for details) to obtain an unbiased estimation to $\braket{\psi|\exp(-itH)|\psi}$\footnote{Throughout this paper, we assume access to $\exp(-itH)$ for any $t\in\mathbb{R}$ for simplicity of presentation. It is straightforward to extend our algorithm to the case with integer powers, where access only to $\exp(-inH)$ for $n\in\mathbb{N}$ is assumed. Refer to \cref{sec:sum} and \cref{sec:intpower} for details.}. Several quantum phase estimation algorithms~\cite{LinTong2022,ni2023lowdepth_2,ni2023lowdepth,Ding2023simultaneous,DingLin2023,Stroeks_2022,wfz22} have been developed assuming access to the Hadamard test circuit. In general, these algorithms, including the one proposed in this paper, involve three steps (refer to \cref{fig:qc} for the flowchart): 1. Generate a proper set of $\{t_n\}^N_{n=1}\subset\mathbb{R}$; 2. Execute the Hadamard test circuits with $t_n$ and obtain the dataset $\{(t_n,Z_n)\}^N_{n=1}$, where $Z_n$ is an approximation of $\braket{\psi|\exp(-it_nH)|\psi}$; 3. Classically post-process $Z_n$ to derive the estimation for $\left\{\lambda_m\right\}_{m\in\mathcal{D}}$. The efficiency of a quantum phase estimation algorithm is then quantified by two metrics: the maximal runtime denoted by $T_{\max}=\max_{1\leq n\leq N}|t_n|$, and the total runtime $T_{\mathrm{total}}=\sum^N_{n=1}|t_n|$. Here, $T_{\max}$ and $T_{\mathrm{total}}$ approximately measure the depth of the circuit and the total cost of the algorithm, respectively. Although the access to the Hamiltonian simulation $\exp(-itH)$ is assumed to be exact in this work for the sake of simplicity, we anticipate that a certain level of simulation error can be allowed as well by treating the simulation error together with the statistical errors and the effect of non-dominant eigenvectors (see \cref{sec:sum} for details).

We further define two types of spectral gaps. The first is the spectral gap between the dominant eigenvalues, denoted by $\Deltam:=\min_{i,j\in \mathcal{D}, j\not=i}\left|\lambda_i-\lambda_j\right|$. The second is the spectral gap between the dominant eigenvalues and the remaining eigenvalues, denoted by $\Delta:=\min_{i\in \mathcal{D}, j\not=i}\left|\lambda_{i}-\lambda_j\right|$. {In this paper, we use the notations $\mathcal{O}_{[a]}$, $\Omega_{[a]}$, and $\Theta_{[a]}$ to indicate that the quantity polynomially depends on the parameters in $ [a] $ if $ a $ is a number greater than one, or inversely proportional to $ [a] $ if $ a $ is a number smaller than one.}

The phase estimation algorithm proposed in this paper exhibits the following properties:
\begin{enumerate}
\item[(1)] Allow imperfect initial state: $\ptail>0.$
\item[(2)] Maintain Heisenberg-limited scaling: Assuming all other parameters remain constant, the algorithm can achieve $\epsilon$-accuracy with $T_{\mathrm{total}}=\widetilde{\mathcal{O}}(1/\epsilon)$.
\item[(3)] No gap requirement: The algorithm can achieve $\epsilon$-accuracy with $T_{\mathrm{total}}=\poly(1/\epsilon)$ for any $\epsilon>0$, where the polynomial and constants are independent of $\Delta,\Deltam$.

\item[(4)]``Short'' depth: When the spectral gap between the dominant eigenvalues $\Deltam>0$ and the precision is small enough, that is, $\epsilon=\widetilde{\mathcal{O}}_{\ptail,p_{\min},|\mathcal{D}|}(\Deltam)$, the maximal runtime $T_{\max}$ can be as small as $\widetilde{\mathcal{O}}_{p_{\min},|\mathcal{D}|}(\ptail/\epsilon)$. Here, the constant before $1/\epsilon$ approaches zero when $\ptail\rightarrow 0$. In addition, the total runtime still achieves the Heisenberg-limited scaling. More specifically, $T_{\rm total}=\widetilde{\mathcal{O}}_{\ptail,p_{\min},|\mathcal{D}|}(1/\epsilon)$.
\end{enumerate}

{Although previous algorithms may fulfill some of the mentioned properties (see \Cref{sec:previous_work} for a detailed discussion), to our knowledge, our algorithm stands out as \textbf{the first algorithm} that can be \textbf{rigorously proven} to simultaneously achieve \textbf{all} four properties, which matches the best available results in the literature.} In particular, the ``short'' circuit depth property is considered important for applications on early fault-tolerant quantum computers~\cite{DingLin2023,KatabarwaGratseaCaesuraEtAl2023}.

The algorithm presented in this paper also fulfills the following two additional properties, enhancing its efficiency compared to others:
\begin{enumerate}
\item[(5)] ``Constant'' depth: When the spectral gap between dominant eigenvalues and all other eigenvalues $\Delta>0$ and $\epsilon=\widetilde{\mathcal{O}}_{p_{\min}}(\Delta)$, we can set $T_{\max}=\widetilde{\mathcal{O}}_{p_{\min}}\left(\delta/\epsilon\log(1/\delta)\right)$ and $T_{\mathrm{total}}=\widetilde{\mathcal{O}}_{p_{\min}}\left(1/(\delta\epsilon)\right)$ for any $\delta=\Omega_{p_{\min}}(\epsilon /\Delta)$. In particular, setting $\delta=\Theta_{p_{\min}}(\epsilon/\Delta)$ gives the constant depth $T_{\max}=\widetilde{\Or}(\Delta^{-1}\log(1 /\epsilon))$, and the total cost is $T_{\mathrm{total}}=\widetilde{\mathcal{O}}_{p_{\min}}\left(\Delta \epsilon^{-2}\right)$.

\item[(6)] The quantum cost of the algorithm  depends \textit{logarithmically} on the number of dominant eigenvalues $|\mc D|$.
\end{enumerate}

Property (5) can be satisfied in some of the prior work, such as~\cite{wfz22,DingLin2023} for ground state energy estimation and~\cite{Ding2023simultaneous} for MEE (see the detailed discussion in \cref{sec:main_theory}). An application of Property (5) is to ensure the algorithm's efficiency in the presence of global depolarizing noise. In this scenario, it is important to maintain a short depth $T_{\max} = \mathcal{O}(\log(1/\epsilon))$ (setting $\delta=\Theta(\epsilon/\Delta)$) to ensure that $T_{\rm total}$ polynomially depends on $1/\epsilon$~\cite{ding2023robust}.

We find that our algorithm does not require meticulous tuning of the simulation parameters. Specifically, for the algorithm to succeed, only knowledge of the upper / lower bounds of the parameters is required (see \cref{sec:idea}).
For optimal complexity, the parameters can be chosen according to $\Delta, \Deltam, \ptail, p_{\min}$ (see \Cref{sec:main_theory} for details).

The rest of this paper is organized as follows. In \cref{sec:idea}, we introduce the main idea of our method, present the algorithm, and provide a brief summary of related works. The complexity of our algorithm is detailed in \cref{sec:main_theory}. The proofs are included in \cref{sec:pf}. \cref{sec:num} proposes several numerical examples and compares our algorithm with previous ones to justify its efficiency. A summary of this work and a discussion of possible extensions are provided in \cref{sec:sum}.

\begin{figure}[H]
\centering
\begin{center}
\includegraphics[width=0.8\textwidth]{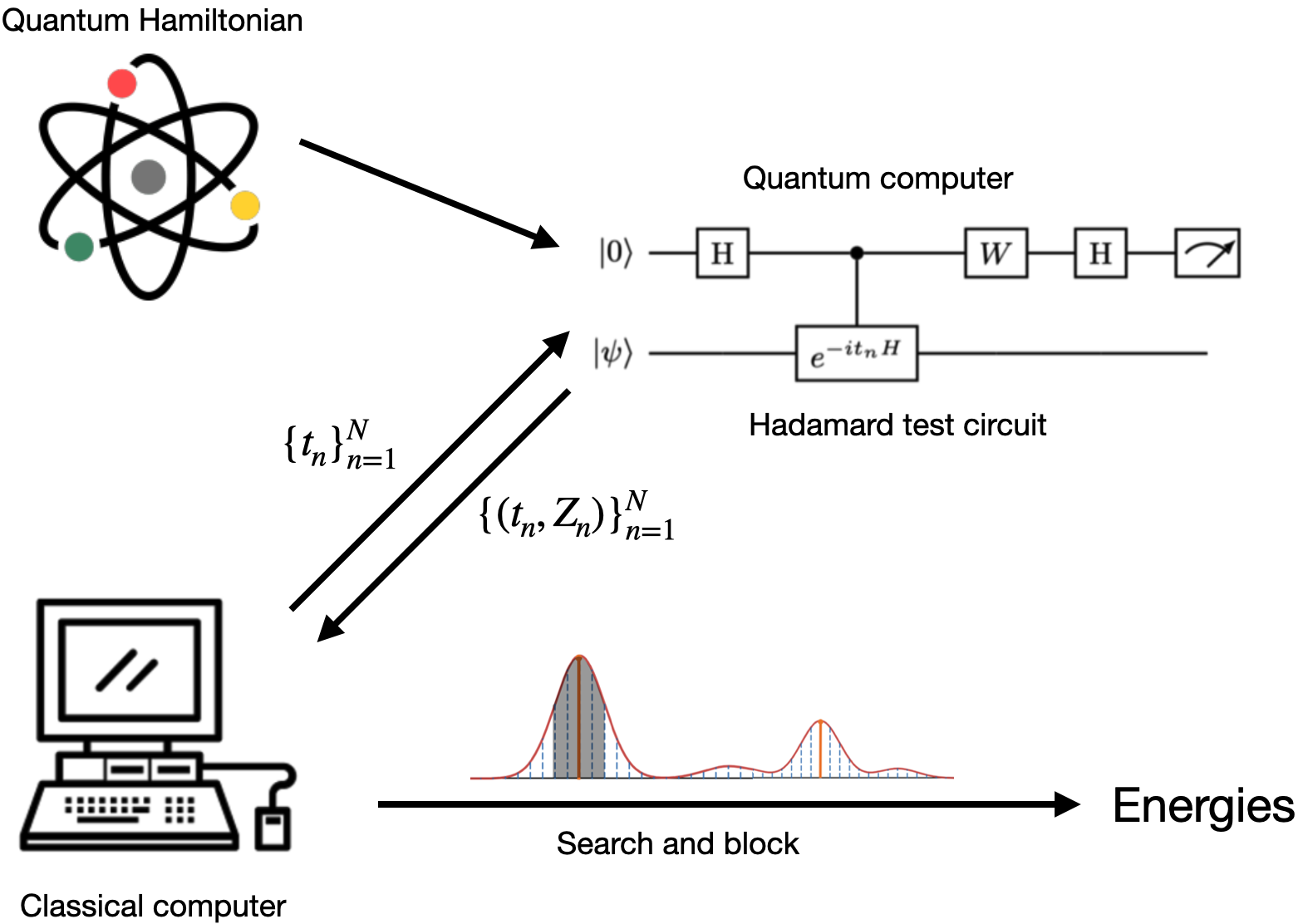}
\end{center}
\caption{Flowchart of the main algorithm. The procedure involves three steps: Firstly, a sequence of $t_n$ is generated from a truncated Gaussian using a classical computer. In the second step, the Hadamard test circuit is implemented on a quantum computer to produce the dataset ${(t_n,Z_n)}$ as defined in \eqref{eqn:dataset}. Within the Hadamard test circuit, we select $W=I$ or $W=S^\dagger$ (where $S$ is the phase gate) to estimate the real or imaginary part of $\braket{\psi|\exp(-itH)|\psi}$. In the final step, postprocessing is performed on the quantum data ${(t_n,Z_n)}$ for eigenvalue estimation.}
\label{fig:qc}
\end{figure}

\section{Main algorithm and previous work}
\label{sec:idea}

We first introduce our main idea informally. Given a Hamiltonian $H$ and an initial state $\ket{\psi}$, our approach relies only on quantum access to the Hadamard test. Specifically, for any $t\in\mathbb{R}$, we can repeat the Hadamard test with $\ket{\psi}$ and $\exp(-iHt)$ several times to obtain an unbiased estimation $Z(t)$ of the following expression:
\begin{equation}\label{eqn:signal}
\mathcal{Z}(t)=\left\langle\psi\right|\exp(-i Ht)\ket{\psi}=\sum^{M}_{m=1}p_m\exp(-i \lambda_m t)\,,
\end{equation}
meaning $\mathbb{E}(Z(t))=\mathcal{Z}(t)$.

The central subroutine of our algorithm involves a filtering-searching process. Given an even probability density $a(t)$, we independently generate $N$ samples $\{t_n\}^N_{n=1}$ from this distribution. Subsequently, we obtain approximations $\{Z(t_n)\}^N_{n=1}$ and calculate the filtering function:
\begin{equation}\label{eqn:filtering_function}
G(\theta)=\left|\frac{1}{N}\sum^N_{n=1}Z(t_n)\exp\left(i\theta t_n\right)\right|\,,
\end{equation}
{where $Z(t_n)$ is a random variable that can take values in $\{\pm 1 \pm i\}$; a detailed discussion can be found later.}

Define the Fourier transform of the probability density function $a(t)$ as follows:
\begin{equation}\label{eqn:fourier_transform_a}
F(x)=\int^\infty_{-\infty}a(t)\exp(ixt)\rd t\,.
\end{equation}
Because $\{t_n\}$ are independently sampled from $a(t)$, it is straightforward to see
\begin{equation}\label{eqn:approximation_of_G}
G(\theta)\approx \left|\mathbb{E}_{t\sim a}\left(Z(t)\exp(i\theta t)\right)\right|=\sum^M_{m=1} p_mF(\theta-\lambda_m)=:\mathcal{G}(\theta)\,,\quad \text{when}\quad N\gg 1\,.
\end{equation}
We observe that $\mathcal{G}(\theta)$ is an even function of real value, given that $a(t)=a(-t)$. In an ideal scenario, assuming that $F(x)$ reaches its maximum at $x=0$ and decays rapidly as $|x|$ increases so that $\max_{1\leq m\leq M-1}\left|F(\lambda_{m+1}-\lambda_m)\right|\ll 1$, we can derive the following approximations:
\begin{equation}\label{eqn:S_t_approximate}
G(\theta)\approx\mathcal{G}(\theta)\approx
\left\{
\begin{aligned}
&p_m F(\theta-\lambda_m),\quad \theta \text{ close to } \lambda_m\,,\\
&0,\quad \text{otherwise}\,.
\end{aligned}\right.
\end{equation}
Subsequently, we can implement the following search procedure to identify all dominant eigenvalues:
\begin{itemize}
\item First, we locate the maximum point $\vtheta_1$ of $G(\theta)$. Following \eqref{eqn:S_t_approximate}, we find $\vtheta_1\approx\lambda_{m_1}$, where $p_{m_1}=\max_{1\leq m\leq M}p_m$.

\item To approximate the next dominant eigenvalue, we establish a block interval $\mathcal{I}_{B,1}=[\vtheta_1-d,\vtheta_1+d]$ around $\vtheta_1$ with a suitable $d$. Subsequently, we identify the second maximal point $\vtheta_2$ of $\mathcal{G}(\theta)$ outside the block interval, denoted {$\vtheta_2=\argmax_{\theta\in\mathcal{I}^c_{B,1}}G(\theta)$.}

Given that $F$ concentrates around $0$, we have $p_{m_1}\exp(\theta-\lambda_{m_1})\ll 1$ when $\theta\in\mathcal{I}^c_{B,1}$. Consequently, the impact of the first dominant eigenvalue is mitigated by blocking the interval in the second search step, allowing us to show $\vtheta_2\approx\lambda_{m_2}$, where $p_{m_2}=\max_{m\neq m_1}p_m$.

\item After acquiring $\vtheta_2$, we update the block interval by defining $\mathcal{I}_{B,2}=[\vtheta_2-d,\vtheta_2+d]\cup \mathcal{I}_{B,1}$ and identify the third maximal point of $G(\theta)$ outside $\mathcal{I}_{B,2}$.

This searching and updating process is iteratively repeated until a set of $|\mathcal{D}|$ ``maximal'' points is discovered. Ultimately, we obtain an approximate set $\left\{\vtheta_{m}\right\}^{|\mathcal{D}|}_{m=1}$ corresponding to the set of dominant eigenvalues $\left\{\lambda_{m}\right\}_{m\in\mathcal{D}}$.
\end{itemize}

To ensure the success of the process, it is crucial to choose a suitable probability density function $a(t)$ such that the function $F(x)$ concentrates around $x=0$. In this paper, we adopt the \emph{truncated Gaussian density function} for $a(t)$:
\begin{equation}\label{eqn:a_T}
\begin{aligned}
 a(t)=&
    \left(1-\int^{\sigma T}_{-\sigma T}\frac{1}{\sqrt{2\pi}T}\exp\left(-\frac{s^2}{2T^2}\right)\textbf{1}_{[-\sigma T,\sigma T]}(s)\rd s\right)\delta_{0}(t)\\
    &+
    \frac{1}{\sqrt{2\pi}T}\exp\left(-\frac{t^2}{2T^2}\right)\textbf{1}_{[-\sigma T,\sigma T]}(t)\,.
\end{aligned}
\end{equation}
{Here, $\delta_0(t)$ is Dirac delta function at point $0$.}
The choice of $a(t)$ is inspired by the fact that the Fourier transform of a Gaussian function remains a Gaussian function and a recent result of spike localization \cite{li2023101577}. The parameter $\sigma$ represents the level of truncation. Specifically, when $\sigma=\infty$, we have $F(x)=\exp\left(-\frac{T^2x^2}{2}\right)$, reaching its maximum at $x=0$ and exponentially decaying to zero with respect to $T|x|$. Furthermore, the use of the truncated Gaussian ensures that the maximum runtime $T_{\max}=\max_n |t_n|$ never exceeds $\sigma T$. In the following part of the paper, we may also employ the notation $a_T(t)=a(t)$ and $F_T(x)=F(x)$ to emphasize the dependence on $T$.

Now, we are ready to introduce our main algorithm. With the motivation explained above, we propose the algorithm in two steps:

\textbf{Step 1: Data generation.}

We implement the Hadamard test quantum circuit as shown in \cref{fig:qc} to obtain our data set. Specifically, we can set $W=I$ (or $W=S^\dagger$ with $S$ being the phase gate), measure the ancilla qubit, and define a random variable $X$ (or $Y$) such that $X=1$ (or $Y=1$) if the outcome is $0$ and $X=-1$ (or $Y=-1$) if the outcome is $1$. Then
\begin{equation}\label{eqn:X_Y}
\mathbb{E}(X+iY)=\left\langle\psi\right|\exp(-i t H)\ket{\psi}\,.
\end{equation}

Given a set of time points $\{t_n\}^N_{n=1}$ drawn from the probability density $a(t)$, we apply Hadamard tests to generate the following data set:
\begin{equation}\label{eqn:dataset}
    \mathcal{D}_{H}=\left\{\left(t_n,Z_n\right)\right\}^{N}_{n=1}:=\left\{\left(t_n,X_n+iY_n\right)\right\}^{N}_{n=1}\,.
\end{equation}
{Here, $X_n$ and $Y_n$ are random variables that each take values of $-1$ or $1$.}
Each evaluation of $X_{n}$ (or $Y_{n}$) only requires running the Hadamard test circuit with $W=I$ (or $W=S^\dagger$)  once at $t=t_n$. 
Referring to \eqref{eqn:X_Y}, we obtain:
{\begin{equation}\label{eqn:Zn_expect}
\mathbb{E}(Z_n)=\left\langle\psi\right|\exp(-i t_n H)\ket{\psi},\quad Z_n\in \left\{\pm 1\pm i\right\},\quad |Z_n|=\sqrt{2}.
\end{equation}}
Hence, $Z_n$ serves as an unbiased and bounded estimate of $\left\langle\psi\right|\exp(-i t_n H)\ket{\psi}$. Furthermore, it should be noted that if we employ the aforementioned method to construct the data set, the maximum simulation time is $T_{\max}=\max_{1\leq n\leq N}|t_n|$, and the total simulation time is $T_{\mathrm{total}}=\sum^N_{n=1}|t_n|$.

We summarize the data generation process in Algorithm \ref{alg:data}.
\begin{breakablealgorithm}
      \caption{Data generator}
  \label{alg:data}
  \begin{algorithmic}[1]
  \State \textbf{Preparation:} Number of data pairs: $N$; Truncated Gaussian density: $a(t)$;
  \State \textbf{Running:}
 \State $n\gets 1$;
  \While{$n\leq N$}
  \State Generate a random variable $t_n$ with the probability density $a(t)$.
  \If{$t_n>0$}
  \State Run the quantum circuit (Figure \ref{fig:qc}) with $t=t_n$ and $W=I$ to obtain $X_{n}$.
  \State Run the quantum circuit (Figure \ref{fig:qc}) with $t=t_n$ and $W=S^\dagger$ to obtain $Y_{n}$.
  \State $Z_{n}\gets X_{n}+i Y_{n}$.
  \EndIf
  \If{$t_n=0$}
  \State {$Z_n\gets 1$.}
  \EndIf
  \State $n\gets n+1$
  \EndWhile
    \State \textbf{Output:} $\left\{(t_n,Z_{n})\right\}^N_{n=1}$
    \end{algorithmic}
\end{breakablealgorithm}
\textbf{Step 2: Filtering and searching.}
After generating the dataset, we define the Gaussian filtering function as in \eqref{eqn:filtering_function}. By selecting a suitable $q$, we define the set of candidates:
\[
\theta_j=-\pi+\frac{jq}{T}\,,\quad 0\leq j\leq J:=\left\lfloor \frac{2\pi T}{q}\right\rfloor\,.
\]
Subsequently, we iterate the previously described searching procedure $K$ times using $G(\theta)$ and the set $\{\theta_j\}^{J}_{j=0}$ to obtain the approximation $\{\vtheta_k\}^K_{k=1}$.

The main algorithm of this paper is described in \cref{alg:main}. It should be noted that our algorithm is flexible in terms of parameter selection. Specifically, a meticulous adjustment of the parameters $\sigma, T, N, \alpha, q$ based on prior knowledge of ${(p_m,\lambda_m)}$ is unnecessary for the algorithm to work. When no prior knowledge is available, we recommend choosing larger values for $\sigma, T, N$. In practice, to avoid detection of repeated dominant eigenvalues during the search, we recommend choosing $q<\alpha\ll T$.  In \cref{sec:main_theory}, we provide a lower bound for $\alpha, \sigma, T, N$ (no required upper bound) and an upper bound for $q$ (no required lower bound) to ensure the algorithm's success.

\begin{breakablealgorithm}
      \caption{Quantum Multiple Eigenvalue Gaussian filtered Search (QMEGS)}
  \label{alg:main}
  \begin{algorithmic}[1]
  \State \textbf{Preparation:} Number of data pairs: $N$; Depth parameter: $T$; Block parameter: $\alpha$; Searching parameter: $q$;
  Truncated Gaussian density: $a_T(t)$; Number of dominant eigenvalues (guess): $K$;
  \State \textbf{Running:}
  \State Generate a data set of size $N$\Comment{Step 1: Generate data}
  \[
   \left\{\left(t_n,Z_{n}\right)\right\}^{N}_{n=1}
  \]
  using \cref{alg:data} with truncated Gaussian density $a_{T}(t)$.
  \State $J\gets \left\lfloor \frac{2\pi T}{q}\right\rfloor$.
  \State Generate discrete candidates: $\theta_j\gets-\pi+\frac{jq}{T}$.
  \State Calculate\Comment{Step 2: Compute the filtered density function}
  \[
  G_j\gets \left|\frac{1}{N}\sum^N_{n=1}Z_n\exp(i\theta_j t_n)\right|\,,\quad 0\leq j\leq J.
  \]
  \State Block set: $\mathcal{I}_{B,1}\gets \emptyset$. \Comment{Step 3: Find peaks}
  \State $k\gets 1$.
  \While{$k\leq K$}
  \State $j_k=\mathrm{argmax}_{\theta_j\notin \mathcal{I}_{B,k}} G_j$.
  \State $\vtheta_k\gets\theta_{j_k}$.
  \State $\mathcal{I}_{B,k+1}\gets \mathcal{I}_{B,k}\cup \left(\vtheta_k-\frac{\alpha}{T},\vtheta_k+\frac{\alpha}{T}\right)$.\Comment{Block interval to avoid finding the same peak}
  \State $k\gets k+1$
  \EndWhile

  \State \textbf{Output:} $\{\vtheta_k\}^K_{k=1}$
    \end{algorithmic}
\end{breakablealgorithm}

Finally, we emphasize that while the informal analysis provided above appears to be conceptually simple, there are still some issues that need to be addressed.
\begin{itemize}
    \item In the informal analysis, to ensure \eqref{eqn:S_t_approximate}, we require \begin{equation}\label{eqn:requirement_1}
\max_{1\leq m\leq M-1}\left|F(\lambda_{m+1}-\lambda_m)\right|\ll 1.
\end{equation}
Given our choice of $a(t)$ as specified in \eqref{eqn:a_T}, where $F(x)\approx\exp\left(-\frac{T^2x^2}{2}\right)$, ensuring that \eqref{eqn:requirement_1} asks for \begin{equation}\label{eqn:requirement_T_bad}
T_{\max}>T>\frac{1}{\min_{1\leq m \leq M-1}|\lambda_{m+1}-\lambda_m|}=\Omega\left(M\right),
\end{equation}
which poses an undesirable requirement in the context of quantum phase estimation\footnote{{For example, in ground state energy estimation, it is well known that $T_{\max}$ can be independent of the spectral gap $\lambda_2-\lambda_1$.}}.
Here, the last equality is based on the assumption that all the eigenvalues of $H$ belong to the interval $[-\pi,\pi]$.

\item In the informal derivation, we ignored the finite sampling and measurement errors and replaced $G(\theta)$ directly with $\mathcal{G}(\theta)$. In practice, however, ensuring small finite sampling and measurement errors requires a large number of samples and is not ideal in quantum computing.
\end{itemize}
In \cref{sec:main_theory}, we present several results that demonstrate that our algorithm can autonomously address the two concerns mentioned above. {We will show that under condition $p_{\min}>\ptail$, \eqref{eqn:S_t_approximate} and \eqref{eqn:requirement_T_bad} are not necessary. Specifically, even if $T$ does not satisfy \eqref{eqn:requirement_T_bad}, and \eqref{eqn:S_t_approximate} does not hold, the set $\{\vtheta_k\}^{|\mathcal{D}|}_{k=1}$ still serves as an approximation to the set of dominant eigenvalues in a meaningful sense.} Moreover, even when $G(\theta)$ is not $\mathcal{O}(\epsilon)$ close to $\mathcal{G}(\theta)$, the search process remains stable enough to ensure the accuracy of the approximation.

We also remark here that the choice of distribution $a(t)$ is not confined to the form specified in \eqref{eqn:a_T}, and it can be tailored based on the setting of the problem. For instance, in scenarios where only a unitary operator $U$ is provided as a black box, and the goal is to retrieve its eigenvalues $e^{i\lambda_m}$, the power $t$ is constrained to integers and thus $a(t)$ needs to be a distribution on integers. A detailed discussion is given in \Cref{sec:intpower}.

\subsection{Comparison with previous work}\label{sec:previous_work}
In this section, we review previous multiple eigenvalue estimation algorithms. A summary of the comparison between different algorithms is listed in \cref{table:1}.

The first work using the Hadamard test circuit to achieve Heisenberg-limited scaling is \cite{LinTong2022}.  Several works have been developed to achieve Heisenberg-limited scaling for multiple eigenvalue estimation and without relying on any gap assumptions (Properties (2) and (3)\footnote{Rigorously speaking, Heisenberg-limited scaling without gap requirement surpasses the strength of Properties (2) and (3) stated in \cref{sec:intro}. It implies that the algorithm can achieve Heisenberg-limited scaling without any gap assumption. Specifically, for any $\epsilon>0$, the algorithm achieves $\epsilon$-accuracy with $T_{\mathrm{total}}=\tilde{\mathcal{O}}_{\ptail,p_{\min}}(1/\epsilon)$, where the constant is independent of $\Delta,
\Deltam$.  As demonstrated in \cref{thm:regime_1}, our algorithm achieves Heisenberg-limited scaling without relying on any gap assumptions.}). For example,~\cite{Dutkiewicz2022heisenberglimited} introduced a method for estimating multiple eigenvalue phases that extends
the idea of robust phase estimation (RPE)~\cite{Higgins_2009,PhysRevA.104.069901,ni2023lowdepth} to multiple eigenvalues
and achieves Heisenberg-limited scaling without relying on any assumptions about the spectral gap. However, the theoretical analysis in~\cite{Dutkiewicz2022heisenberglimited} is based on the assumption that all non-dominant modes vanish, as defined in~\cite[Definition 3.1 and Theorem 4.5]{Dutkiewicz2022heisenberglimited}, specifically $p_{\mathrm{tail}}=0$.
This drawback has been resolved by a recent work~\cite{ni2023lowdepth_2}. The robust multiple phase estimation (RMPE) method in \cite{ni2023lowdepth_2} extends the Kitaev-type RPE method~\cite{PhysRevA.104.069901,ni2023lowdepth} to the multiple eigenvalue estimation problem by using adaptive simulation time amplifying factors and suitable signal processing algorithms. In contrast to~\cite{Dutkiewicz2022heisenberglimited}, the algorithm proposed in~\cite[Section III]{ni2023lowdepth_2} can achieve Heisenberg-limited scaling without the requirement of a gap between dominant eigenvalues, even when dealing with an imperfect initial state ($\ptail>0$) (Properties (1)-(3)). The design of the algorithm is based on a dedicated line spectrum estimation algorithm \cite{li2023101577}, which employs the same Gaussian filtering function as employed in this work. Moreover, in scenarios where a gap exists between dominant eigenvalues, the algorithm presented in \cite[Section V]{ni2023lowdepth_2} combines the line spectrum estimation algorithm with ESPRIT to estimate the dominant eigenvalues with a short circuit depth (Property (4)). Specifically, to achieve $\epsilon$-accuracy, the algorithm requires $T_{\max}=\mathcal{O}_{p_{\min},|\mathcal{D}|}(\max\{\ptail^{-1}\Deltam^{-1}, \ptail \epsilon^{-1}\})$ and $T_{\mathrm{total}}=\mathcal{O}_{p_{\min},|\mathcal{D}|}(\ptail^{-2}\Deltam^{-1}+\ptail^{-1} \epsilon^{-1})$, resulting in a short circuit depth as long as $\epsilon = \mathcal{O}(\ptail^2\Deltam)$.
Compared to the algorithms outlined in~\cite[Sections III, V]{ni2023lowdepth_2}, the quantum cost of QMEGS exhibits a logarithmic dependence on $|\mathcal{D}|$, and QMEGS can further achieve the ``constant'' depth property in the presence of a spectral gap between dominant and other eigenvalues. In addition, QMEGS does not require meticulous parameter adjustment (e.g., $\sigma, T, N, \alpha, q$) based on prior knowledge of $\{(p_m,\lambda_m)\}$, and can be more flexible than the algorithms proposed in~\cite{ni2023lowdepth_2}.

An optimization-based signal processing method called the multi-modal, multi-level quantum complex exponential least squares (MM-QCELS) method has recently been proposed ~\cite{Ding2023simultaneous}, which generalizes the quantum complex exponential least squares (QCELS) method~\cite{DingLin2023} to the setting of multiple eigenvalues. When $\Deltam>0$, MM-QCELS can approximate these dominant eigenvalues with Heisenberg limited scaling, ``short'', or ``constant'' circuit depth (Properties (1), (2), (4), (5)). Furthermore, the quantum cost of MM-QCELS only depends logarithmically on $|\mathcal{D}|$ (Property (6)).  To achieve Heisenberg-limited scaling, MM-QCELS needs to generate a sequence of datasets on quantum computers. The classical optimization procedure of MM-QCELS to find $K$ dominant eigenvalues solves an optimization problem in a $K$-dimensional space (see \cref{sec:summary_numerical_method} for a brief overview). In the worst-case scenario, the classical post-processing cost can grow exponentially in $\epsilon^{-K}$.
Compared to MM-QCELS, QMEGS has a much simpler data generation and searching process. The algorithm only requires a single dataset generated by a single $T$. Additionally, leveraging block intervals in the searching process, QMEGS can achieve $\epsilon$-accuracy with Heisenberg-limited scaling even when $\epsilon=\Omega(\Deltam)$. Regarding the classical processing cost, QMEGS only requires evaluating the filter function at a finite number of discrete points in $[-\pi,\pi]$ and the number of evaluations is $\mathcal{O}(1/\epsilon)$ and is independent of $K$.

Quantum subspace diagonalization (QSD), quantum Krylov, matrix pencil, and ESPRIT methods, as highlighted in studies such as~\cite{Cortes2021QuantumKS,Huggins_2020,PRXQuantum.3.020323,PhysRevA.95.042308,Motta_2020,Parrish2019QuantumFD,PRXQuantum.2.010333,doi:10.1021/acs.jctc.9b01125,OBrienTarasinskiTerhal2019,doi:10.1137/21M145954X,shen2023estimating}, offer an alternative way to solve the eigenvalue estimation problem. These methods estimate eigenvalues by addressing specific projected eigenvalue problems or singular value problems and have proven valuable for estimating ground-state and excited-state energies across various scenarios. Despite classical perturbation theories suggesting potential challenges, such as ill-conditioned projected problems and sensitivity to noise, empirical observations indicate that these quantum methods often outperform pessimistic theoretical predictions. Recently,
\cite{Stroeks_2022} proposed and investigated the complexity of a quantum phase estimation algorithm based on ESPRIT \cite{rk89, 9000636}. It recovers the frequencies by computing the SVD of the Hankel matrix generated from samples of the signal. Assuming a sufficiently large dominant spectral gap $\Delta_{\rm dom}$ and $\ptail=0$, \cite{Stroeks_2022} establishes that ESPRIT can achieve $\epsilon$ precision with $T_{\max}=\wt{\cal O}(\mathrm{poly}(|\mathcal{D}|)\epsilon^{-o(1)})$ and $T_{\rm total}=\wt{\cal O}(\mathrm{poly}(|\mathcal{D}|)(p_{\min}\epsilon)^{-2})$. Notably, the relationship $T_{\rm total}=\Omega(T_{\max}^2)$, which arises from the aliasing issue, poses a challenge to the ESPRIT method (as well as other ESPRIT-like methods, such as the observable dynamical mode decomposition (ODMD) method \cite{shen2023estimating}) in achieving the Heisenberg-limited scaling (see \cref{sec:summary_numerical_method} for detail).
More recently,~\cite[Section IV]{ni2023lowdepth_2} presented a multilevel ESPRIT technique to overcome the restriction of $\ptail=0$ and reach the Heisenberg limited scaling with short circuit depth (Properties (1), (2), (4)). Specifically, when $\epsilon\ll \Deltam$, the authors illustrate that multilevel ESPRIT can achieve $\epsilon$-accuracy with $T_{\max}=\widetilde{\mathcal{O}}(\mathrm{poly}(|\mathcal{D}|)\ptail(p_{\min}\epsilon)^{-1})$ and $T_{\mathrm{total}}=\widetilde{\mathcal{O}}(\mathrm{poly}(|\mathcal{D}|)(\ptail p_{\min}\epsilon)^{-1})$. In addition, despite the apparent importance of the gap assumption in the previous analysis of ESPRIT~\cite{Stroeks_2022,ni2023lowdepth_2}, our numerical experiments in \cref{sec:num} suggest that this assumption might be relaxed in real-world applications.

It should be noted that filtering techniques have also been explored in previous studies. For example, in~\cite{wfz22}, the authors introduced the Gaussian derivative filter to estimate the ground-state energy (smallest eigenvalues) and achieve a ``Constant'' depth property (Property (5)).
However, it remains unclear how to extend their algorithms to address the MEE problem while satisfying the first four properties (Properties (1)-(4)). More recently, ~\cite{PRXQuantum.4.040341} proposed an algorithm that employs a filter function based on the derivative Heaviside function. In this approach, the eigenvalues of $H$ can be approximated by identifying the local maxima of the approximate derivative of the cumulative distribution function (CDF).
As $T_{\max}$ tends to infinity, their approach aligns with QMEGS in that both methods aim to identify the local maximal point in the sum of delta functions.
However, when aiming to identify dominant eigenvalues with finite $T_{\max}$, the performance of their algorithm and its ability to achieve Heisenberg-limited scaling are unclear.

\begin{table}[h!]
\centering
\begin{tabular}{c|ccccc}
\hline
\hline
\textbf{Algorithms} & \multicolumn{4}{c}{\textbf{Properties}} & \textbf{Comments} \\
 & Allow & Heisenberg & No gap & ``Short'' &  \\
 & $\ptail>0$ & limit & requirement & depth &  \\ \hline
QEEA~\cite{Somma2019} & \cmark & \xmark & {\cmark} & \xmark &  \\ \hline
ESPRIT~\cite{Stroeks_2022} & ? & \xmark & ? & \xmark &  \\ \hline
\cite{Dutkiewicz2022heisenberglimited} & ? & \cmark & \cmark & \xmark &  \begin{tabular}[c]{@{}c@{}} $\poly(|{\cal D}|)$ quantum cost\end{tabular} \\ \hline
\cite[Theorem III.5]{ni2023lowdepth_2} & \cmark & \cmark & \cmark & \xmark &  \multirow{2}{*}{$\poly(|{\cal D}|)$ quantum cost} \\ \cline{1-5}
\cite[Theorem V.1]{ni2023lowdepth_2} & \cmark & \cmark & \xmark & \cmark &  \\ \hline
MM-QCELS~\cite{Ding2023simultaneous} & \cmark & \cmark & \xmark & \cmark & \begin{tabular}[c]{@{}c@{}}``Constant'' depth,\\ $\log |{\cal D}|$ quantum cost \\ large classical cost\end{tabular} \\ \hline
QMEGS (this work) & \cmark & \cmark & \cmark & \cmark & \begin{tabular}[c]{@{}c@{}}``Constant'' depth, \\ $\log |{\cal D}|$ quantum cost\end{tabular} \\ \hline
\end{tabular}
\caption{Comparison of the existing theoretical analysis of different methods for multiple eigenvalue estimation. {
In the table, a question mark indicates that the method potentially satisfies this property, although there is no existing result that rigorously proves it.}
For simplicity, certain properties in this table represent a slightly relaxed version of those introduced in \cref{sec:intro}. Specifically, Allow $\ptail>0$ means that the algorithm can deal with an imperfect initial state.
The Heisenberg limit implies that the algorithm can achieve $\epsilon$-accuracy with $T_{\mathrm{total}}=\widetilde{\mathcal{O}}(1/\epsilon)$, assuming that all other parameters remain constant. No gap requirement means that the algorithm can achieve $\epsilon$-accuracy with $T_{\mathrm{total}}=\poly(1/\epsilon)$, where the polynomial and constants are independent of $\Delta,\Deltam$. ``Short'' depth  means that, when $\epsilon=\widetilde{\mathcal{O}}_{\ptail,p_{\min},|\mathcal{D}|}(\Deltam)$, the algorithm can achieve $\epsilon$-accuracy with $T_{\max}=\widetilde{\mathcal{O}}_{p_{\min},|\mathcal{D}|}\left(\ptail/\epsilon\right)$ and $T_{\rm total}=\widetilde{\mathcal{O}}_{\ptail,p_{\min},|\mathcal{D}|}\left(1/\epsilon\right)$. A more detailed comparison can be found in \cref{sec:previous_work}.
}
\label{table:1}
\end{table}

\section{Statement of the main results}\label{sec:main_theory}
This section introduces the complexity result of \cref{alg:main}. The goal is to demonstrate that, across all three regimes categorized by our knowledge of the spectral gap, the algorithm consistently identifies a highly accurate approximation of the dominant eigenvalues through appropriate parameter selection. We provide the proof of the results in \cref{sec:pf}.

The three cases are divided as follows depending on the values of $T,\Delta,\Deltam$:
\begin{itemize}
    \item General case: In this case, the theoretical result does not require assumptions about $T$, $\Deltam$, and $\Delta$. Consequently, we refer to it as the \emph{No gap requirement} regime. The result is summarized in the following theorem.
    \begin{theorem}[$\forall T>0$]\label{thm:regime_1} Assume $p_{\min}>\ptail$ and $K\ge|\mathcal{D}|$. Given the probability of failure $\eta>0$, we choose the following parameters:
\begin{itemize}
    \item Block constant: $\alpha=\Omega\left(\log^{1/2}\left(\frac{1}{p_{\min}-\ptail}\right)\right)$,
    \item Searching parameter: $q=\mathcal{O}\left(\log^{1/2}\left(\frac{p_{\min}}{\ptail+(p_{\min}-\ptail)/2}\right)\right)$, $q<\alpha/3$, and $\alpha/q\in\mathbb{N}$,
    \item Truncation parameter: $\sigma=\Omega\left(\log^{1/2}\left(\frac{1}{p_{\min}-\ptail}\right)\right)$,
    \item Number of samples: $N=\Omega\left(\frac{1}{(p_{\min}-\ptail)^2}\log\left(\left(\frac{T}{q}+|\mathcal{D}|\right)\frac{1}{\eta}\right)\right)$.
\end{itemize}
Then, with probability at least $1-\eta$, we have that for each $i\in\mathcal{D}$, there exists $1\leq k_i\leq |\mathcal{D}|$ such that
\begin{equation}\label{eqn:distance_bound_regime_1}
\left|\lambda_i-\vtheta_{k_i}\right|\leq \frac{\alpha}{T}\,.
\end{equation}
In particular, for any $\epsilon>0$, to achieve
\[
\left\{\lambda_m\right\}_{m\in\mathcal{D}}\subset \cup_{k}[\vtheta_k-\epsilon,\vtheta_k+\epsilon]\,,
\]
it suffices to choose
\[
T_{\max}=\widetilde{\Theta}\left(\frac{1}{\epsilon}\right),\quad T_{\mathrm{total}}=\widetilde{\Theta}\left(\frac{1}{(p_{\min}-\ptail)^2\epsilon}\log\left(\frac{|\mathcal{D}|}{\eta}\right)\right)\,,
\]
\end{theorem}

The above theorem guarantees that Properties (1)-(3) are attained as outlined in \cref{sec:intro}. This result shares similarities with \cite[Theorem III.5]{ni2023lowdepth_2}. Since a spectral gap is not assumed, $T$ may not be large enough to distinguish between close dominant eigenvalues, and it is possible for two dominant eigenvalues to fall within the same interval $[\vtheta_k-\alpha/T,\vtheta_k+\alpha/T]$.
Also, some intervals may not include any dominant eigenvalues since $K \geq |\mathcal{D}|$. But \eqref{eqn:distance_bound_regime_1} ensures that the dominant eigenvalues are covered by the union set $\cup_{k}[\vtheta_k-\alpha/T,\vtheta_k+\alpha/T]$.

We briefly describe the strategy of proving the theorem using proof by contradiction as follows: First, assume that there exists $\lambda_{m^\star}$ such that $\lambda_{m^\star}\notin\overline{\mathcal{I}}_{B,|\mathcal{D}|}$ (here $\overline{\mathcal{I}}$ denotes the closure of $\mathcal{I}$). Then one can show that there exists a grid point $\vtheta^\star\notin\mathcal{I}_{B,|\mathcal{D}|}$ such that $|\vtheta^\star-\lambda_{m^\star}|\leq q/T$ using the condition of $q,\alpha$. Thus the value of $G$ at $\vtheta^\star$ has the following lower bound due to \eqref{eqn:approximation_of_G}:
\[
G(\vtheta^\star)\geq p_{\min}\exp\left(-\frac{q^2}{2}\right)+\text{small error}\,.
\]
Then each $\left[\vtheta_k-\frac{\alpha}{3T},\vtheta_k+\frac{\alpha}{3T}\right]$ must cover some dominant eigenvalue $\lambda_m$ since otherwise
\[
G(\vtheta_k)\leq \ptail+\exp\left(-\frac{\alpha^2}{18}\right)+\text{small error}< p_{\min}\exp\left(-\frac{q^2}{2}\right)+\text{small error}\leq G(\vtheta^\star)\,,
\]
which contradicts with the fact that $\vtheta_k$ is the maximal point at $k$-th step. Finally, since $\lambda_{m^\star}$ is not covered by any block intervals, some dominant eigenvalue $\lambda_{m'}$ must be covered by at least two such intervals $\left[\vtheta_k-\frac{\alpha}{3T},\vtheta_k+\frac{\alpha}{3T}\right]$ and $\left[\vtheta_j-\frac{\alpha}{3T},\vtheta_j+\frac{\alpha}{3T}\right]$ ($j\neq k$). However, according to the definition of the block interval, we have $|\vtheta_k-\vtheta_j|\geq \frac{\alpha}{T}$, which implies $\left[\vtheta_k-\frac{\alpha}{3T},\vtheta_k+\frac{\alpha}{3T}\right]\cap \left[\vtheta_j-\frac{\alpha}{3T},\vtheta_j+\frac{\alpha}{3T}\right]=\emptyset$. This contradicts the existence of $\lambda_{m'}$ and concludes the proof.

In \cref{sec:pf}, we account for both random and truncation errors and form a rigorous proof using this idea.

\item Gapped dominant eigenvalues case: In this case, we assume that $T$ is large enough to allow the Gaussian filter to distinguish dominant eigenvalues. We show that Algorithm \ref{alg:main} attains Heisenberg-limited scaling with $T_{\max}=\ptail/(p_{\mathrm{min}}\epsilon)$. Similar results have been attained by previous algorithms, such as MM-QCELS~\cite{Ding2023simultaneous} and \cite[Theorem IV.2]{ni2023lowdepth_2}, and are  referred to as \emph{Short depth} regime (Property (4)). The results are summarized in the following theorem:
\begin{theorem}[$T\gg \Deltam^{-1}$]\label{thm:regime_2} Assume that $p_{\min}>\ptail$ and define \[\zeta=\min\{(p_{\min}-\ptail)/\ptail,1\}\,.\] Given failure probability $\eta>0$, we choose $q<\alpha/3$, $\alpha/q\in\mathbb{N}$ such that
\begin{equation}\label{eqn:condition:para_regime_2}
\sigma=\Omega\left(\log^{1/2}\left(\frac{1}{\zeta\ptail}\right)\right),\quad \alpha=\Omega\left(\sigma\right),\quad q=\mathcal{O}\left(\frac{\zeta\ptail}{\sigma}\right)\,.
\end{equation}
and
\[
T=\Omega\left(\frac{\alpha}{\Deltam}\right),\quad N=\Omega\left(\frac{1}{(\zeta\ptail)^{2}}\log\left(\left(\frac{T}{q}+|\mathcal{D}|\right)\frac{1}{\eta}\right)\right)\,.
\]
There exists
\[
Q=\Theta\left(\exp\left(\Theta(\zeta^{-1})\right)(\sigma\zeta+1)\frac{\ptail}{p_{\min}T}\right)
\]
such that, with probability $1-\eta$, for each $i\in\mathcal{D}$, there exists a unique $1\leq k_i\leq |\mathcal{D}|$ such that
\begin{equation}\label{eqn:distance_bound}
\left|\lambda_i-\vtheta_{k_i}\right|\leq \frac{Q}{T}\,.
\end{equation}
Furthermore, we have $k_i\neq k_j$ for $i\neq j$. In particular, when $p_{\min}>(1+o(1))\ptail$, for $\epsilon=\widetilde{\mathcal{O}}\left(\ptail\Deltam/p_{\min}\right)$,
the following quantum cost is enough to achieve $\max_{m\in\mathcal{D}}|\lambda_m-\vtheta_{k_m}|\leq \epsilon$:
\[
T_{\max}=\widetilde{\Theta}\left(\frac{\ptail}{p_{\min}\epsilon}\right),\quad T_{\mathrm{total}}=\widetilde{\Theta}\left(\frac{1}{\ptail p_{\min}\epsilon}\log\left(\frac{|\mathcal{D}|}{\eta}\right)\right)\,.
\]
\end{theorem}

In contrast to the general case, in this scenario,
\[
\exp\left(-T^2(\lambda_{m_1}-\lambda_{m_2})^2/2\right)\ll 1, \quad m_1,m_2\in\mathcal{D}.
\]
This condition ensures that the Gaussian filter can effectively distinguish between different dominant eigenvalues, establishing a one-to-one correspondence between dominant eigenvalues $\lambda_m$ and their corresponding maximal points $\vtheta_{k_m}$.

To establish the short circuit depth result in \eqref{eqn:distance_bound}, we begin by noting that, neglecting the influence of other eigenvalues and noise, we must have $|\lambda_m-\vtheta_{k_m}|\leq q/T$ due to the decreasing property of the function $F(x)$. When accounting for the impact of other eigenvalues, the following observations hold: 1. The influence of other dominant eigenvalues is negligible when determining $\vtheta_k$ since $T=\widetilde{\Omega}\left(\Deltam^{-1}\right)$; 2. The effect of the tail eigenvalues diminishes as $\ptail$ approaches zero since $\left|\sum_{m\in\mathcal{D}^c}p_m\exp(-(\theta-\lambda_m)^2T^2/2)\right|\leq \ptail$.
By leveraging these observations and controlling the impact of random noise, we establish the validity of \eqref{eqn:distance_bound} in \cref{sec:pf}.

\item The case where a gap is known between dominant eigenvalues and the tail: In this case, we assume that $T\gg \Delta^{-1}$. The Gaussian filter thus becomes sufficiently sharp to distinguish between dominant eigenvalues and tail eigenvalues. In this scenario, we demonstrate that \cref{alg:main} can achieve Heisenberg-limited scaling, with $T_{\max}\epsilon$ arbitrarily close to zero. Furthermore, \cref{alg:main} is capable of achieving $\epsilon$ accuracy with $T_{\max}=\mathcal{O}(\log(1/\epsilon))$ and $T_{\mathrm{total}}=\mathcal{O}(1/\epsilon^2)$. Since the maximum running time $T_{\max}$ depends only logarithmically on $\epsilon$, we refer to this regime as the \emph{``Constant'' depth} regime (Property (5)). It should be noted that similar results have been achieved previously in the task of ground-state energy estimation algorithms, such as QCELS~\cite{DingLin2023,ding2023robust}, and \cite{wfz22}. Here, we extend this result to the case of multiple eigenvalue estimation.

\begin{theorem}[$T\gg \Delta^{-1}$]\label{thm:regime_3} Assume that $p_{\min}>\ptail$. Choose all parameters such that they satisfy the condition of \cref{thm:regime_1}. Furthermore, given any $0<\zeta<p_{\min}/4$ and failure probability $\eta>0$, we choose
\[
\sigma=\Omega\left(\log^{1/2}\left(\frac{1}{\zeta}\right)\right), \quad \alpha=\Omega(\sigma),\quad q=\mathcal{O}\left(\frac{\zeta}{\sigma}\right)\,,
\]
and
\[
T=\Omega\left(\frac{\alpha}{\Delta}\right),\quad N=\Omega\left(\frac{1}{\zeta^{2}}\log\left(\left(\frac{T}{q}+|\mathcal{D}|\right)\frac{1}{\eta}\right)\right)\,.
\]
Let $Q=\mathcal{O}\left(\frac{\zeta\sigma}{p_{\min}}\right)$. Then, with high probability, for each $i\in\mathcal{D}$, there exist $1\leq k_i\leq |\mathcal{D}|$ such that
\begin{equation}\label{eqn:distance_bound_regime_3}
\left|\lambda_i-\vtheta_{k_i}\right|\leq \frac{Q}{T}\,.
\end{equation}
Furthermore, it holds that $k_i\neq k_j$ for $i\neq j$. In particular, for $\epsilon=\widetilde{\mathcal{O}}\left(\Delta/p_{\min}\right)$ and $\delta=\Omega(p_{\min}\epsilon/\Delta)$, the following quantum cost is enough to achieve $\max_{m\in\mathcal{D}}|\lambda_m-\vtheta_{k_m}|\leq \epsilon$:
\[
T_{\max}=\Theta\left(\frac{\delta}{p_{\min}\epsilon}\log\left(\frac{1}{\delta}\right)\right),\quad T_{\mathrm{total}}=\widetilde{\Theta}\left(\frac{1}{p_{\min}\delta\epsilon}\log\left(\frac{|\mathcal{D}|}{\eta}\right)\right)\,.
\]
\end{theorem}
According to the above theorem, we can increase the number of samples $N$ to decrease $T_{\max}\epsilon$ to an arbitrarily small value. However, it is crucial to acknowledge the tradeoff, as an increase in the number of samples $N$ also amplifies the total running time. In the extreme scenario where $N^{-1/2}=\widetilde{\Theta}(\epsilon)$, \cref{alg:main} achieves $\epsilon$ accuracy with $T_{\max}=\Theta\left(\log\left(1/\epsilon\right)\right)$.

The proof strategy employed in \cref{thm:regime_3} closely mirrors that of \cref{thm:regime_2}. The key difference lies in the influence of tail eigenvalues. As $T=\Omega(\alpha/\Delta)$, we observe that \[\left|\sum_{m\in\mathcal{D}^c}p_m\exp(-(\vtheta-\lambda_m)^2T^2/2)\right|\leq \exp(-\alpha^2/2)=\mathcal{O}(\zeta^2)\,.\]
This implies that the impact of tail eigenvalues can be arbitrarily small. By combining this observation with the proof strategy applied in \cref{thm:regime_2}, we can establish the validity of \eqref{eqn:distance_bound_regime_3}.

The algorithm's ability to achieve a ``Constant'' depth comes from the rapid decay of the filter function $F(x)$ as $|x|$ increases. A similar concept is employed in QCELS~\cite{DingLin2023,ding2023robust}, where solving the optimization problem is nearly equivalent to a filtering and searching process using the Gaussian filter function. It is also important to note that the choice of the Gaussian filter function might not be exclusive to achieve the ``Constant" depth property. For instance, in the algorithm presented in~\cite{wfz22}, the filtering and searching process is applied using the Gaussian derivative function to approximately locate the 'zero point' of the filtered signal, which corresponds to the \emph{ground state} energy. Exploring the possibility of extending other filter functions to MEE and achieving similar results would be an interesting avenue for further investigation.

\end{itemize}

\section{Numerical experiments}\label{sec:num}

In this section, we demonstrate the efficiency of QMEGS numerically by comparing it with previously established quantum phase estimation algorithms: MM-QCELS~\cite{Ding2023simultaneous}, QPE (textbook version~\cite{NielsenChuang2000}), and ESPRIT~\cite{Stroeks_2022}. A brief summary of these three methods is provided in Appendix \ref{sec:summary_numerical_method}. We share the code on Github (\url{https://github.com/zhiyanding/phase_estimation_methods}).

We consider three models: 1. Toy Hamiltonian with almost zero dominant spectral gap between dominant eigenvalues (\cref{sec:artificial}); 2. TFIM model (\cref{sec:Ising}); 3. Hubbard model (\cref{sec:Hubbard}). In all cases, we normalize the Hamiltonian spectrum so that the eigenvalues lie within the range of $[-\pi/4,\pi/4]$ for our numerical experiments. Specifically, we use the normalized Hamiltonian in the experiment.
\begin{equation}\label{eqn:normalize_H}
\widetilde{H}=\frac{\pi H}{4\|H\|_2}\,.
\end{equation}
In our test, we construct an initial state $\ket{\psi}$ with $p_1=p_2=0.4$. Therefore, we let $\lambda_1,\lambda_2$ be the dominant eigenvalues and set $\mathcal{D}=\{1,2\}$. In addition, the dominant spectral gap is {$\Deltam=\lambda_2-\lambda_1$}. To test the four algorithms, we set $T=100\times 2^n$ with $n=1,2,\cdots,7$. For QMEGS (\cref{alg:main}), we choose $N=500$, $K=2$, $\alpha=5$, $\sigma=1$, $q=0.05$.
For MM-QCELS~\cite[Algorithm 2]{Ding2023simultaneous}, the parameters are set to $K=2,T_0=100,N_0=10^3,N_{j\geq 1}=500$ and $\sigma=1$. For ESPRIT (see \cref{sec:summary_numerical_method}), we set $N=\lfloor T\rfloor$ and $K=2$. {For QPE, we consider the specific version designed for estimating the ground state energy $\lambda_1$, which is a variation of the version proposed in ~\cite[Chapter 5.2]{NielsenChuang2000}. Specifically, we sample its distribution $N_{\mathrm{QPE}}=\lceil 6/p_1\rceil$ times, and take the minimal of the measures as the approximation to $\lambda_{1}$ (see \cref{sec:summary_numerical_method}). The error QPE is defined as the error for estimating this single eigenvalue.}

We then use QMEGS (\cref{alg:main}), MM-QCELS, and ESPRIT to estimate the two dominant eigenvalues and measure the max-min error:
\begin{equation}\label{eqn:max_error}
\mathrm{error}=\max_{m\in \mathcal{D}}\min_{\sigma\in\{1,2,\cdots,K\}^{\otimes K}}|\vtheta^*_{\sigma_i}-\lambda_m|\,.
\end{equation}
In our case, $D=\{1,2\},K=2$. This testing criterion gives QPE with an advantage. Nonetheless, we illustrate that {the performance of QMEGS, MM-QCELS, and ESPRIT can all to surpass that of QPE.}

\subsection{Almost zero dominant spectral gap}\label{sec:artificial}
In this test, we randomly generate a Hamiltonian $H$
\[
H=\sum^M_{m=1}\lambda_m\ket{\psi_m}\bra{\psi_m}\,,
\]
with dimension $M=20$, $\|H\|=1$, and $\lambda_1-\lambda_0=10^{-3}$. Thus, we have $\Deltam=\lambda_1-\lambda_0=10^{-3}$. We then randomly generate an initial state $\ket{\psi}$ such that $p_1=p_2=0.4$.

We apply QMEGS (\cref{alg:main}), MM-QCELS, ESPRIT, and QPE to estimate the dominant eigenvalues ($\lambda_1,\lambda_2$) of the normalized Hamiltonian $\widetilde{H}$ according to \eqref{eqn:normalize_H}.  A comparison of the results is shown in \cref{fig:artificial}.

We observe that the QPE error consistently decreases with increasing $T_{\max}$. However, with the same $T_{\max}$, the error for QPE is consistently higher than that of the other three methods. For QMEGS (\cref{alg:main}) and ESPRIT, errors do not change much with small values of $T_{\max}$ and gradually decrease when $T_{\max}$ is large enough. This behavior aligns with expectations, as for small $T_{\max}$, QMEGS and ESPRIT struggle to distinguish between two dominant eigenvalues. When $T_{\max}$ becomes sufficiently large, these algorithms automatically differentiate between $\lambda_1$ and $\lambda_2$, leading to more accurate estimates. In contrast, the performance of MM-QCELS is different. Despite a large $T_{\max}$, MM-QCELS struggles to differentiate between two dominant eigenvalues, resulting in the error persisting around $\Delta_m=10^{-3}$. Additionally, when comparing the total evolution time ($T_{\mathrm{total}}$), QMEGS (\cref{alg:main}) is shown to be more time efficient than both QPE and ESPRIT.

\begin{figure}[htbp]
     \subfloat{
         \centering
         \includegraphics[width=0.5\textwidth]{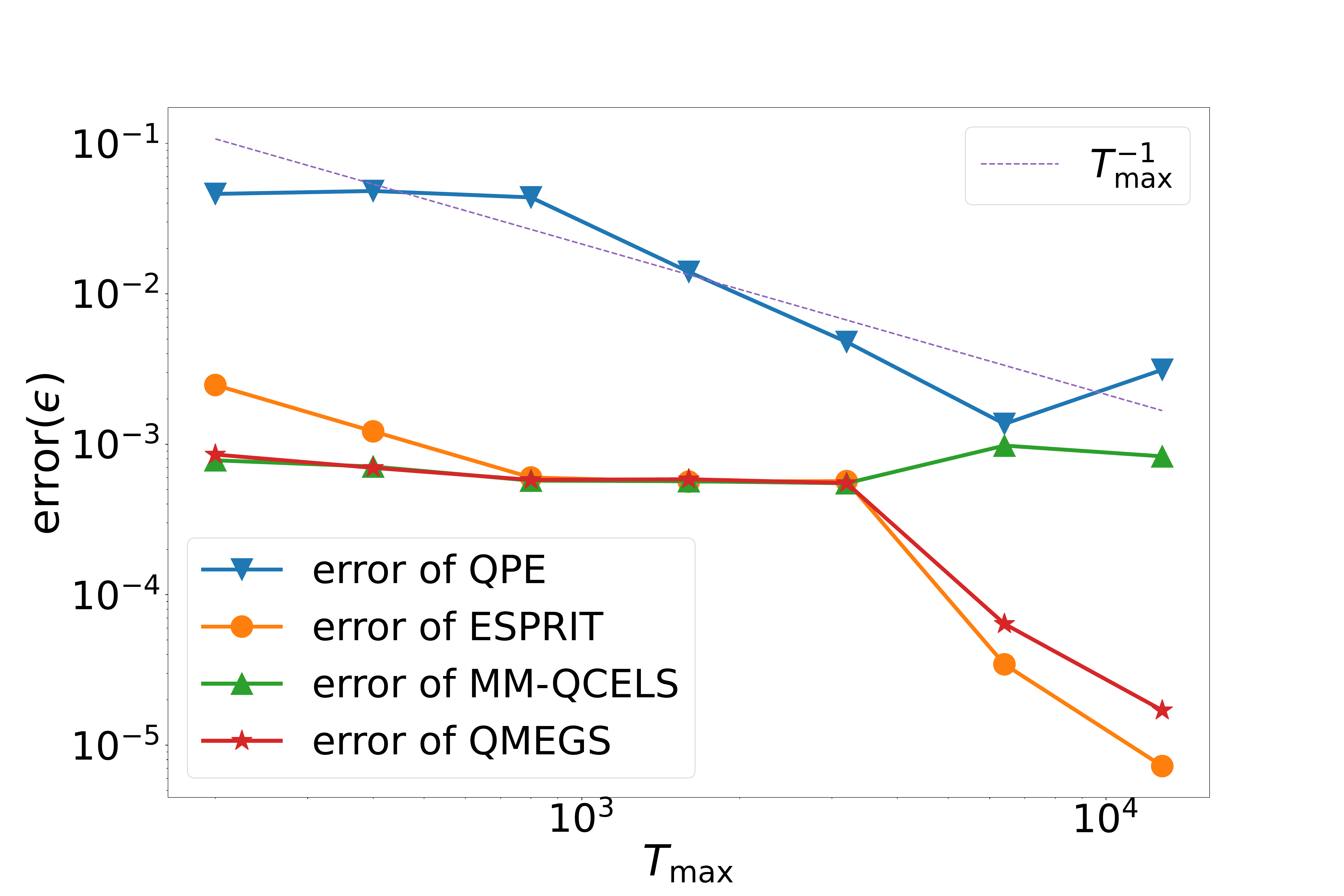}
     }
     \subfloat{
         \centering
         \includegraphics[width=0.5\textwidth]{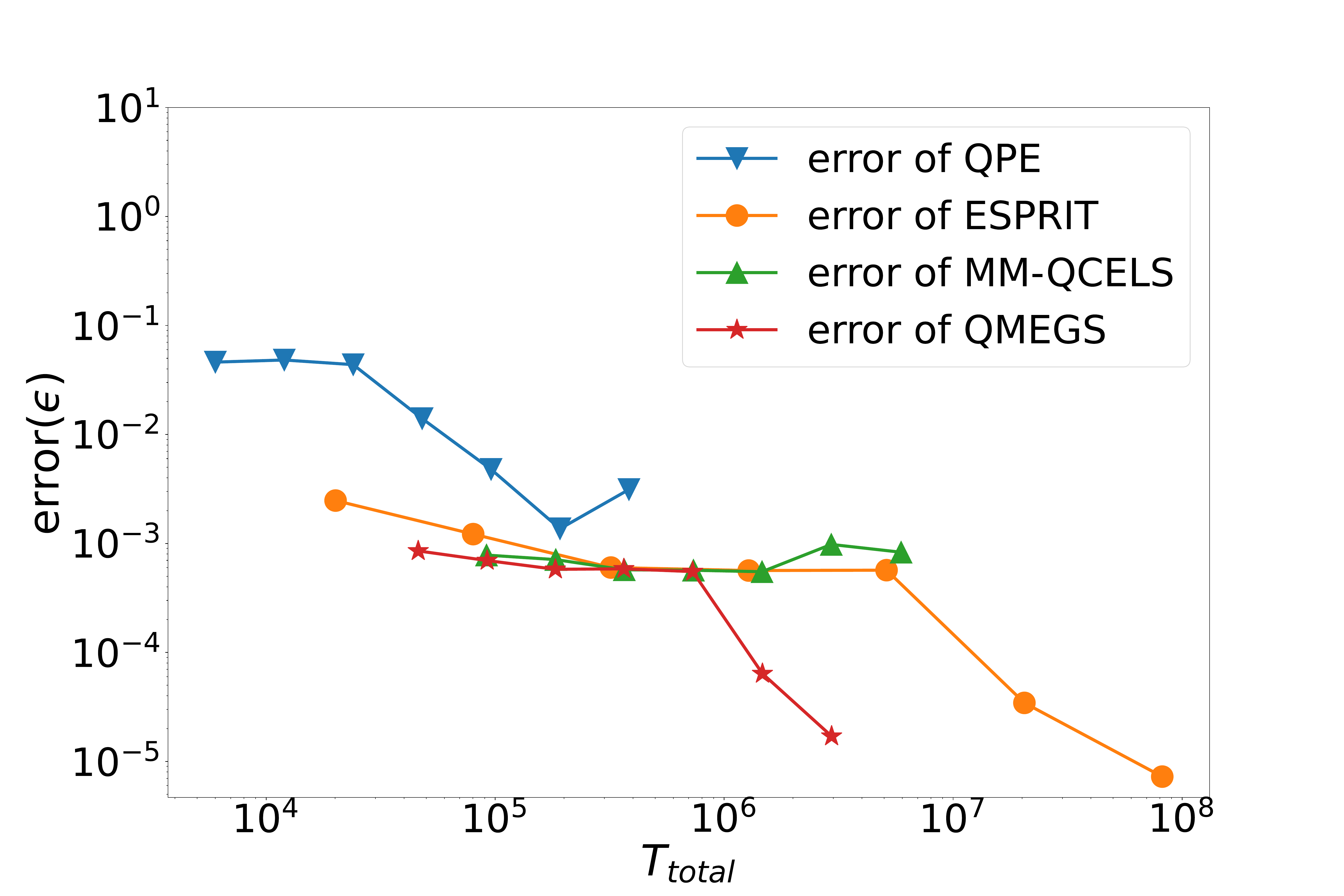}
     }
     \caption{
     \label{fig:artificial} QMEGS (\cref{alg:main}) vs. MM-QCELS vs. ESPRIT vs. QPE in almost zero dominant spectral gap Hamiltonian with $p_1=p_2=0.4$. We highlight that when the spectral gap is small ($\Delta_m=10^{-3}$), MM-QCELS struggles to differentiate between two prominent eigenvalues even when $T$ is large. In contrast, QMEGS and ESPRIT exhibit the ability to automatically distinguish between $\lambda_1$ and $\lambda_2$, achieving high accuracy as $T$ increases.}
\end{figure}

\subsection{Ising model}\label{sec:Ising}
Consider the one-dimensional transverse field Ising model (TFIM) model defined on $L$ sites with periodic boundary conditions:
\begin{equation}\label{eqn:H_Ising}
H=-\left(\sum^{L-1}_{i=1} Z_{i}Z_{i+1}+Z_{L}Z_1\right) -g\sum^L_{i=1} X_i
\end{equation}
where $g$ is the coupling coefficient, $Z_i,X_i$ are Pauli operators for the $i$-th site and the dimension of $H$ is $2^L$. We choose $L=8,g=4$ and apply four algorithms to estimate the dominant eigenvalues ($\lambda_1,\lambda_2$) of the normalized Hamiltonian $\widetilde{H}$ according to \eqref{eqn:normalize_H}. We note that in this case, $\Deltam=\lambda_1-\lambda_0\approx 0.2$. A comparison of the results is shown in \cref{fig:Ising}.

We observe that the errors of QMEGS, MM-QCELS, and QPE exhibit a proportionality to the inverse of the maximal evolution time ($T_{\max}$). In particular, the constant factor $\delta=T\epsilon$ for QMEGS and MM-QCELS is considerably smaller than that of QPE. As illustrated in \cref{fig:Ising}, QMEGS and MM-QCELS significantly reduce the maximal evolution time by two orders of magnitude. {In contrast to the other three methods, the error of ESPRIT decreases with $T^{-1.5}_{\max}$\footnote{{A very recent work~\cite{ding2024esprit} rigorously proves this based on a highly nontrivial matrix perturbation argument. Specifically, when $T_{\max} \gg \Delta_m $, the authors show that the error of ESPRIT scales as $ T_{\max}^{-1.5} $.}}, enabling it to achieve a smaller error than the other methods when $T_{\max}$ is sufficiently large.} However, in terms of the total evolution time $T_{\mathrm{total}}$, ESPRIT incurs a cost that is one order of magnitude higher than QMEGS.

\begin{figure}[htbp]
     \subfloat{
         \centering
         \includegraphics[width=0.48\textwidth]{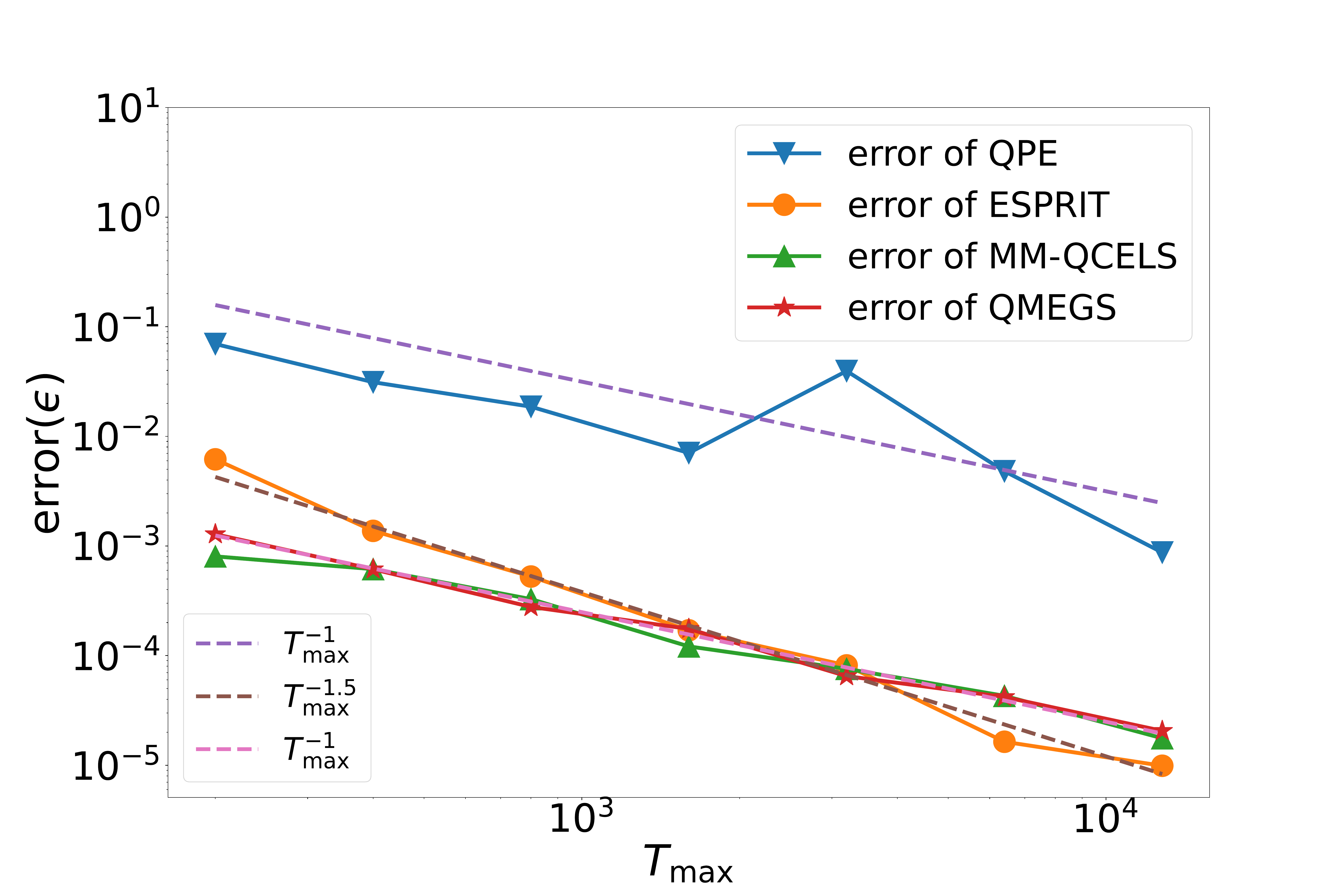}
     }
     \subfloat{
         \centering
         \includegraphics[width=0.48\textwidth]{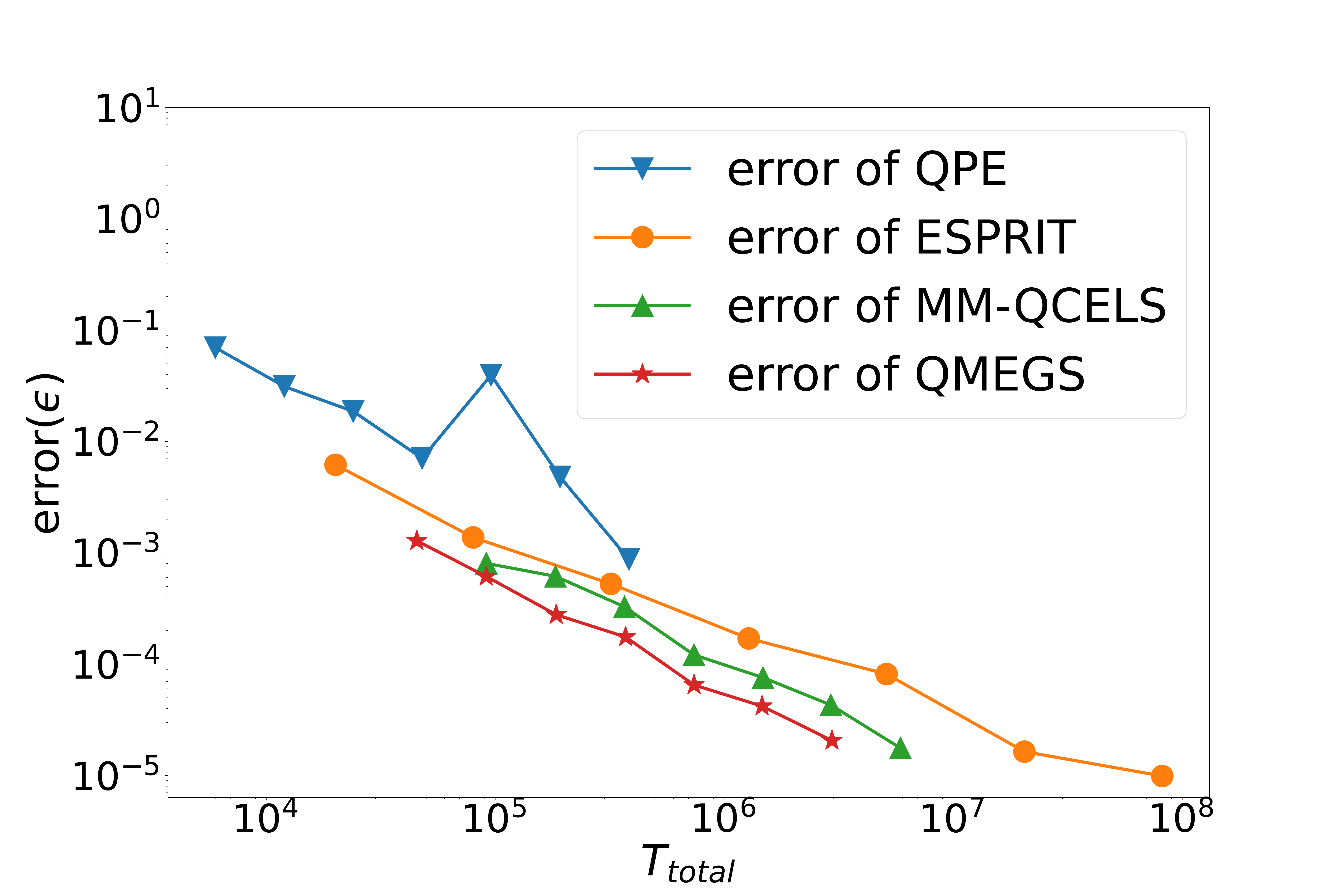}
     }
     \caption{
     \label{fig:Ising} QMEGS (\cref{alg:main}) vs. MM-QCELS vs. ESPRIT vs. QPE in TFIM model with 8 sites with $p_1=p_2=0.4$. Left: Depth ($T_{\max}$); Right: Cost ($T_{\mathrm{total}}$). The errors of QPE, QMEGS, and MM-QCELS methods exhibit a linear scaling with $1/T_{\max}$, where the constant factor $\delta=T\epsilon$ for QMEGS and MM-QCELS is notably smaller than that of QPE. On the other hand, ESPRIT's error scales as $T_{\max}^{-1.5}$. In terms of the total evolution time $T_{\mathrm{total}}$, QMEGS stands out with the smallest $T_{\mathrm{total}}$ compared to other methods and is shown to be one order of magnitude more cost-effective than ESPRIT.}
\end{figure}

\subsection{Hubbard model}\label{sec:Hubbard}
Consider the one-dimensional Hubbard model defined on $L$ spinful sites with open boundary conditions:
\[
H=-t\sum^{L-1}_{j=1}\sum_{\sigma\in\{\uparrow,\downarrow\}}c^\dagger_{j,\sigma}c_{j+1,\sigma}+U\sum^L_{j=1}\left(n_{j,\uparrow}-\frac{1}{2}\right)\left(n_{j,\downarrow}-\frac{1}{2}\right).
\]
Here $c_{j,\sigma}(c^\dagger_{j,\sigma})$ denotes the fermionic annihilation (creation) operator on the site $j$ with spin $\sigma$, $n_{j,\sigma}=c^\dagger_{j,\sigma}c_{j,\sigma}$ denotes the number operator for $\sigma\in \{\uparrow,\downarrow\}$, and $t,U\in \R$ are parameters. We choose $L=4,8$, $t=1$, $U=10$ and test four algorithms to estimate $\lambda_1,\lambda_0$. We note that the dominant spectral gap $\Deltam=\lambda_1-\lambda_0\approx 0.02$.

The numerical results are presented in \cref{fig:Hubbard_8}. Similar to the TFIM model, the maximal evolution time of QMEGS and MM-QCELS is nearly two orders of magnitude smaller than that of QPE, while the error of ESPRIT scales as $T_{\max}^{-1.5}$. Additionally, the overall computational cost of QMEGS is lower than that of other methods and is almost one order of magnitude lower than the cost of ESPRIT.

\begin{figure}[htbp]
     \centering
     \subfloat{
         \centering
         \includegraphics[width=0.48\textwidth]{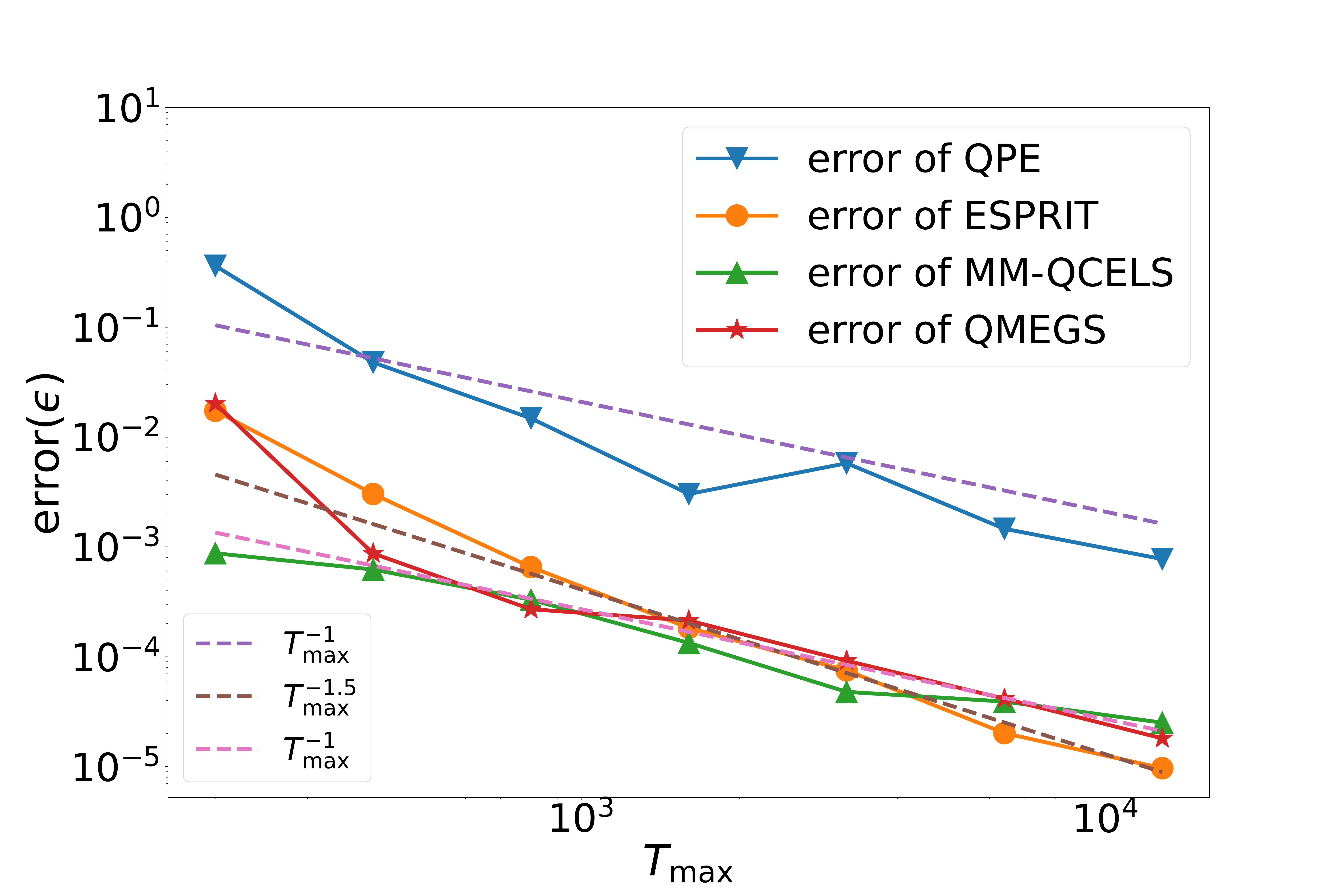}
     }
     \hfill
     \subfloat{
         \centering
         \includegraphics[width=0.48\textwidth]{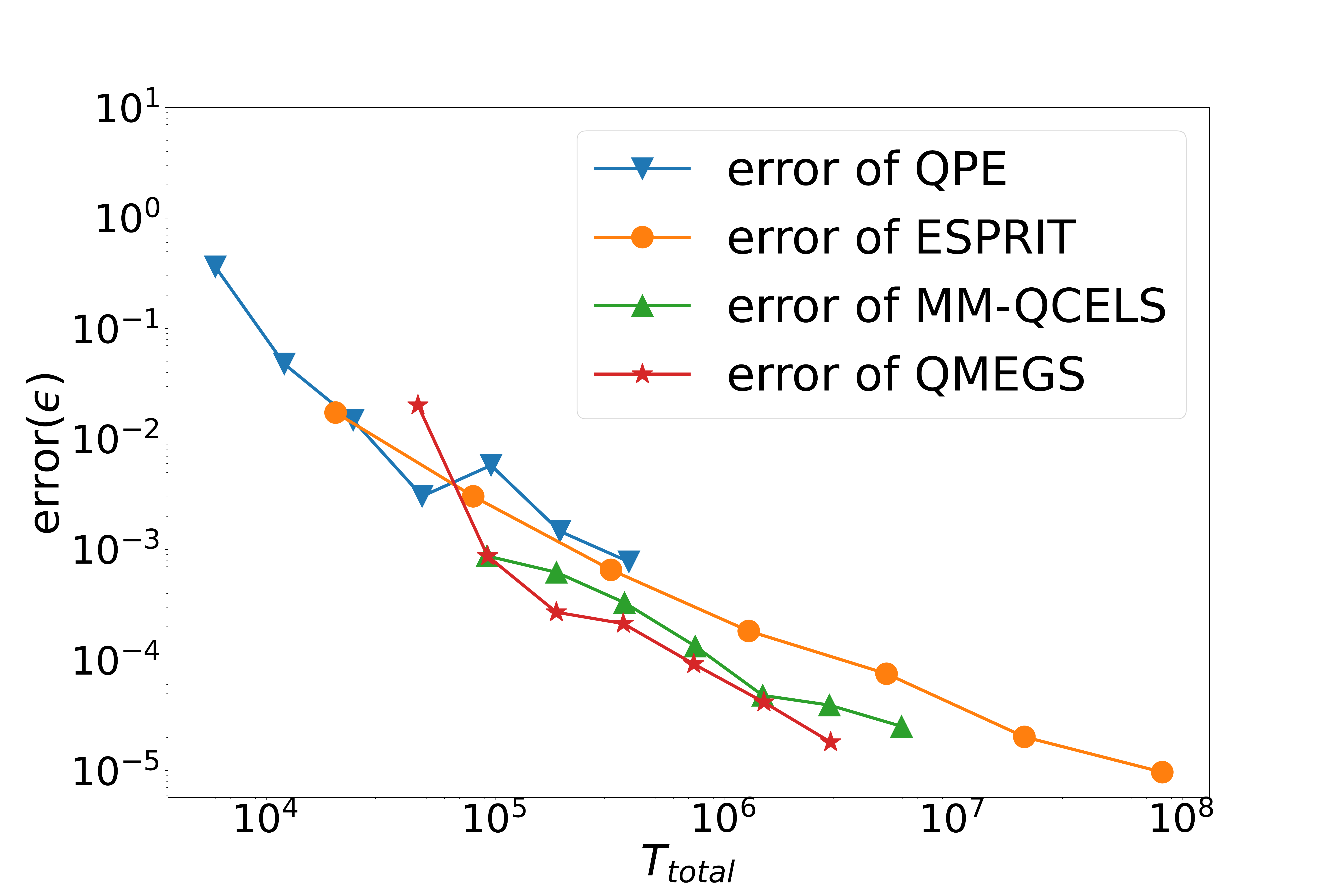}
     }
     \caption{
     \label{fig:Hubbard_8}  QMEGS (\cref{alg:main}) vs. MM-QCELS vs. ESPRIT vs. QPE in TFIM model with 8 sites with $p_1=p_2=0.4$. Left: Depth ($T_{\max}$); Right: Cost ($T_{\mathrm{total}}$). QMEGS and MM-QCELS methods exhibit significantly smaller $T_{\max}$ values compared to QPE. Moreover, QMEGS achieves the smallest total evolution time and is one order of magnitude more cost-effective than ESPRIT.}
\end{figure}

\section{Discussions}\label{sec:sum}
In this paper, we introduce a novel algorithm, Quantum Multiple Eigenvalue Gaussian filtered Search (QMEGS), designed for addressing the multiple eigenvalues estimation problem. In contrast to preceding algorithms for multiple eigenvalue estimation, QMEGS is the first algorithm capable not only of consistently achieving Heisenberg-limited scaling but also of attaining ``short'' and ``constant'' depth in the presence of gaps between eigenvalues. Additionally, QMEGS offers flexibility in parameter selection. It is not imperative to meticulously set the parameters based on prior knowledge of ${(p_m,\lambda_m)}$ for the algorithm to operate effectively. Numerous numerical results substantiate the efficiency and flexibility of the algorithm.

There are several directions to extend this work.
\begin{itemize}
    \item While QMEGS has small quantum computational complexity, the classical cost of the algorithm is linearly dependent on the inverse of the precision $\epsilon^{-1}$ and the size of the domain (assuming $\lambda_{m}\subset [-\pi,\pi]$). Exploring more efficient search algorithms is an intriguing direction to reduce this linear dependence to logarithmic potentially.

    \item QMEGS demonstrates flexibility in parameter selection, achieving accurate estimations when $\alpha, \sigma, T, N, q^{-1}$ are chosen to be sufficiently large. The challenge of adaptively selecting these parameters to achieve optimal complexity without prior knowledge of ${(p_m,\lambda_m)}$ remains an open problem.

    \item
In this study, we assume the exactness in the Hamiltonian simulation $\exp(-iHt)$ for simplicity. However, this assumption may not hold in the context of early fault-tolerant quantum computers. In cases where the Hamiltonian simulation is not exact, $Z_n$ becomes a biased estimate for $\braket{\psi|\exp(-it_nH)|\psi}$. Investigating the impact of different noise models on QMEGS performance is an intriguing problem, and the results depend on the noise models chosen and the assumptions about $Z_n$. For example, when the bias size is independent of $t$ and sufficiently small, such as in the case of small additive noise, a similar analysis to the proof of \cref{thm:regime_1} (\cref{sec:pf}) can be employed. This allows us to absorb the error into statistical error and non-dominant parts, demonstrating that the method still achieves Heisenberg-limited scaling.

\item
It is important to emphasize that while ESPRIT-type algorithms \cite{Stroeks_2022,shen2023estimating} may not achieve Heisenberg-limited scaling and theoretically require gap assumptions, these algorithms still perform effectively in the absence of gaps between dominant eigenvalues and show rapid error decay with $T_{\max}$. In the presence of a gap, numerical results indicate that the decay rate of the error appears to be $\mathcal{O}(T^{-1.5}_{\max})$. To our knowledge this phenomenon was not previously reported in the literature.  The theoretical underpinning of these phenomena remains an open question, which makes it an intriguing direction for further exploration into the complexity of ESPRIT-type algorithms.

\item As highlighted in \cref{sec:previous_work}, the Gaussian filtering function is not the only option for addressing the MEE problem. Some other filtering functions, such as the Gaussian derivative filter~\cite{wfz22} and derivative Heaviside function~\cite{PRXQuantum.4.040341}, might also be used to solve the MEE problem. The exploration of alternative filtering functions that offer improved complexity performance and flexibility remains an open problem.
\end{itemize}

\vspace{1em}
\textbf{Acknowledgment}

This material is partially supported by the U.S. Department of Energy, Office of Science, National Quantum Information Science Research Centers, Quantum Systems Accelerator (Z.D.), by the Applied Mathematics Program of the US Department of Energy (DOE) Office of Advanced Scientific Computing Research under contract number DE-AC02-05CH1123 and the Simons Investigator program (L.L.), by the National Science Foundation under awards DMS-2011699 and DMS-2208163 (L.Y.), and by DOE Grant No. DE-SC0024124 (R.Z.).
We thank the Institute for Pure and Applied Mathematics
(IPAM) for its hospitality throughout the semester-long program   ``Mathematical and Computational Challenges in Quantum Computing'' in Fall 2023 during which this work was initiated.

\bibliographystyle{quantum}
\bibliography{ref}

\appendix

\section{Proofs}\label{sec:pf}

In this section, we prove \cref{thm:regime_1,thm:regime_2,thm:regime_3}. We first introduce a lemma that helps us bound the error caused by time truncation and random noise. Define the error term
\[
E_j=|E\left(\theta_j\right)|,\quad
E\left(\theta\right)=:\frac{1}{N}\sum^N_{n=1}Z_n\exp(i\theta t_n)-\sum^M_{m=1}p_m\exp\left(-\frac{T^2(\lambda_m-\theta)^2}{2}\right)\,,
\]
where $\theta_j$ is defined in \Cref{alg:main}.
We note that $E_j$ contains errors caused by time truncation and finite samples of $t_n$. We demonstrate the smallness of $E_j$ in the following lemma:
\begin{lemma}\label{lem:truncation_error} Given $\delta>0$ and the pair of overlap-eigenvalues $\left\{\left(p_i,\lambda_i\right)\right\}^M_{m=1}$, if $\sigma=\Omega\left(\log^{1/2}\left(1/\delta\right)\right)$, we have
\begin{equation}\label{eqn:time_truncation_error}
\left|\mathbb{E}\left(\frac{1}{N}\sum^N_{n=1}Z_n\exp(i\theta t_n)\right)-\sum^M_{m=1}p_m\exp\left(-\frac{T^2(\lambda_m-\theta)^2}{2}\right)\right|\leq \delta^2\,.
\end{equation}
Furthermore, define $E\left(\theta\right)=\frac{1}{N}\sum^N_{n=1}Z_n\exp(i\theta t_n)-\sum^M_{m=1}p_m\exp\left(-\frac{T^2(\lambda_m-\theta)^2}{2}\right)$.
Given $\eta>0$ and $\Theta=\left\{\theta_j\right\}^J_{j=1}\cup \left\{\lambda_m\right\}_{m\in\mathcal{D}}$, if $N=\Omega\left(\frac{1}{\delta^2}\log\left(\left(J+|\mathcal{D}|\right)\frac{1}{\eta}\right)\right)$,  we have
\begin{equation}\label{eqn:random_error}
\mathbb{P}\left(\max_{\theta\in\Theta}\left|E(\theta)\right|\leq \delta\right)\geq 1-\eta\,,
\end{equation}
and
\begin{equation}\label{eqn:random_error_2}
\mathbb{P}\left(\cap_{\theta,\theta'\in\Theta,\theta\neq \theta'}\left\{\left|E(\theta)-E(\theta')\right|\leq \sigma T\delta|\theta-\theta'|+\delta^2\right\}\right)\geq 1-\eta\,.
\end{equation}
\end{lemma}
\begin{proof}
According to \cref{alg:data}, when $\sigma=\Omega\left(\log^{1/2}\left(1/\delta\right)\right)$, we have
\begin{equation}
\begin{aligned}
&\left|\mathbb{E}\left(\frac{1}{N}\sum^N_{n=1}Z_n\exp(i\theta t_n)\right)-\sum^M_{m=1}p_m\exp\left(-\frac{T^2(\lambda_m-\theta)^2}{2}\right)\right|\\
\le&  \int_{|t|>\sigma T}\frac{1}{\sqrt{2\pi} T}e^{-\frac{t^2}{2T^2}}\left|\sum^M_{m=1}p_me^{-it\la_m}e^{i\theta t}\right|\d t=2\int_{\sigma T}^\infty\frac{1}{\sqrt{2\pi} T}e^{-\frac{t^2}{2T^2}}\d t\\
=&\sqrt{\frac{2}{\pi}}\int_{\sigma}^\infty e^{-\frac{s^2}{2}}\d s< e^{-\sigma^2}<\delta^2,
\end{aligned}
\end{equation}
where the inequality $\int_{x}^\infty e^{-\frac{1}{2}t^2}\d t \le \sqrt{\frac{\pi}{2}}e^{-\frac{x^2}{2}}$ is used in the second last inequality. This proves \eqref{eqn:time_truncation_error}. The proof of \eqref{eqn:random_error} and \eqref{eqn:random_error_2} is the same as the proof of \cite[Appendix B.3 Lemma 4 eqn. (B8)]{ding2023robust}, thus, we omit it here.
\end{proof}
Recall that
\[
G(\theta) = \left|\frac{1}{N}\sum^N_{n=1}Z_n\exp(i\theta t_n)\right|\,.
\]
We are ready to prove \cref{thm:regime_1}.
\begin{proof}[Proof of \cref{thm:regime_1}]
According to \cref{lem:truncation_error} \eqref{eqn:random_error} with $\delta=\mathcal{O}(p_{\min}-\ptail)$ and $J=\mathcal{O}(T/q)$,  we obtain
\begin{equation}\label{eqn:random_error_regime_regime_1}
\mathbb{P}\left(\max_{\theta \in\Theta}\left|E(\theta)\right|<\frac{p_{\min}-\ptail}{8}\right)\geq 1-\eta\,.
\end{equation}
Thus, to prove Theorem \ref{thm:regime_1}, it suffices to show \eqref{eqn:distance_bound_regime_1} assuming $E_j<\frac{p_{\min}-\ptail}{8}$. Under this condition, we consider two classes of candidates:
\begin{itemize}
    \item When $\theta_j\in \cup_{m\in\mathcal{D}}\left[\lambda_m-\frac{q}{T},\lambda_m+\frac{q}{T}\right]$, we obtain
    \[
    \begin{aligned}
    G_j=&~G(\theta_j)\geq \left|\sum^M_{m=1}p_m\exp\left(-\frac{T^2(\lambda_m-\theta_j)^2}{2}\right)\right|-E_j\\
    > &~p_{\min}\exp\left(-\frac{q^2}{2}\right)-\frac{p_{\min}-\ptail}{8}\\
    \geq &~ \ptail+\frac{p_{\min}-\ptail}{2}-\frac{p_{\min}-\ptail}{8}\geq \ptail+\frac{3(p_{\min}-\ptail)}{8}\,.
    \end{aligned}
    \]
    where we use $E_j<\frac{p_{\min}-\ptail}{8}$ in the second inequality, and the condition of $q$ in the last inequality.

    \item When $\theta_j\notin \cup_{m\in\mathcal{D}}\left[\lambda_m-\frac{\alpha}{3T},\lambda_m+\frac{\alpha}{3T}\right]$, we obtain
    \[
    \begin{aligned}
    G_j\leq &~\left|\sum_{m\in\mathcal{D}}p_m\exp\left(-\frac{T^2(\lambda_m-\theta_j)^2}{2}\right)\right|+\left|\sum_{m\in\mathcal{D}^ c}p_m\exp\left(-\frac{T^2(\lambda_m-\theta_j)^2}{2}\right)\right|+E_j\\
    \leq &~\exp\left(-\frac{\alpha^2}{18}\right)+\ptail+\frac{p_{\min}-\ptail}{8}\leq \ptail+\frac{p_{\min}-\ptail}{4}\,,
    \end{aligned}
    \]
    where $\mathcal{D}^c:=[M]\backslash \mathcal{D}$, and we use $E_j<\frac{p_{\min}-\ptail}{8}$ in the second inequality and the condition of $\alpha$ in the last inequality.
\end{itemize}
Because $q<\alpha/3$, the above two inequalities imply
\begin{equation}\label{eqn:separation_regime_1}
\max_{\theta_j\notin \cup_{m\in\mathcal{D}}\left[\lambda_m-\frac{\alpha}{3T},\lambda_m+\frac{\alpha}{3T}\right]}G_j<\min_{\theta_j\in \cup_{m\in\mathcal{D}}\left[\lambda_m-\frac{q}{T},\lambda_m+\frac{q}{T}\right]} G_j\,.
\end{equation}
Now, we prove \eqref{eqn:distance_bound_regime_1} using the proof by contradiction argument. Assuming that there exists $m^\star\in\mathcal{D}$ such that
\[
\lambda_{m^\star}\notin \cup^{K}_{k=1}\left[\vtheta_k-\frac{\alpha}{T},\vtheta_k+\frac{\alpha}{T}\right]\,,
\]
then the grid point $\vtheta^*$ closest to $\lambda_{m^\star}$ is not in the block set since $\alpha/q\in\mathbb{N}$ and the block set is a union of open intervals.
According to \eqref{eqn:separation_regime_1}, we must have $\vtheta_k\in \cup_{m\in\mathcal{D}, m\not=m^\star}\left[\lambda_m-\frac{\alpha}{3T},\lambda_m+\frac{\alpha}{3T}\right]$ for all $1\leq k\leq |\mathcal{D}|$. By the pigeonhole principle, there exist $k_1<k_2$ and $m'\in\mathcal{D}$ such that $\left\{\vtheta_{k_1},\vtheta_{k_2}\right\}\subset \left[\lambda_{m'}-\frac{\alpha}{3T},\lambda_{m'}+\frac{\alpha}{3T}\right]$. Thus, $|\vtheta_{k_1}-\vtheta_{k_2}|\le \frac{2\alpha}{3T}$. However, this contradicts the blocking process that ensures
\[
\vtheta_{k_2}\notin \left(\vtheta_{k_1}-\frac{\alpha}{T},\vtheta_{k_1}+\frac{\alpha}{T}\right)\,.
\]
This concludes the proof.
\end{proof}

Next, we give the proof of \cref{thm:regime_2}. Define the tail term that contains the impact of non-dominant eigenvalues:
\begin{equation}\label{eqn:G_tail}
G_{\mathrm{tail}}(\theta)=\sum_{m\in\mathcal{D}^c}p_m\exp\left(-\frac{T^2(\lambda_m-\theta)^2}{2}\right)\,.
\end{equation}
The proof of \cref{thm:regime_2} is as follows.
\begin{proof}[Proof of \cref{thm:regime_2}]
Because $p_{\min}>p_{\mathrm{tail}}$, it is straightforward to see that the condition of \cref{thm:regime_1} is satisfied. In addition, using the condition of $N$ with \cref{lem:truncation_error} by setting $\delta=\zeta\ptail/10$ and $J=\mathcal{O}(T/q)$,
we obtain
\begin{equation}\label{eqn:error_small_1_regime_2}
\mathbb{P}\left(\max_{\theta\in\Theta}\left|E(\theta)\right|\leq \zeta\ptail/10\right)\geq1-\eta/2\,,
\end{equation}
and
\begin{equation}\label{eqn:error_small_1_regime_3}
\mathbb{P}\left(\cap_{\theta,\theta'\in\Theta,\theta\neq \theta'}\left\{\left|E(\theta)-E(\theta')\right|=\sigma T\zeta\ptail|\theta-\theta'|/10+\zeta^2\ptail^2/100\right\}\right)\geq 1-\eta/2\,.
\end{equation}
where $\Theta=\left\{\theta_j\right\}^J_{j=1}\cup \left\{\lambda_m\right\}_{m\in\mathcal{D}}$.

To prove the theorem, it suffices to show that when \eqref{eqn:error_small_1_regime_2} and \eqref{eqn:error_small_1_regime_3} hold, \eqref{eqn:distance_bound} holds.
Fixed $1\leq k\leq |\mathcal{D}|$. Define $\mathcal{D}_{k_-}:=\{m\in\mathcal{D}:\lambda_m\in \mathcal{I}_{B,k}\}$, where $\mathcal{I}_{B,k}$ is the block interval defined in the algorithm. We note that $\mathcal{D}_{k_-}$ is the set of indices that have already been covered by the blocked set. Because $T=\Omega\left(\frac{\alpha}{\Deltam}\right)$, we find that each interval $\left[\vtheta_j-\frac{\alpha}{T},\vtheta_j+\frac{\alpha}{T}\right]$ contains exactly one dominant eigenvalue for all $1\leq j\leq k-1$ and ${\cal D}\backslash \mathcal{D}_{k_-}\neq \emptyset$. Then, to prove \eqref{eqn:distance_bound}, it suffices to show that, if $\left|\vtheta_j-\lambda_{m_j}\right|\leq \frac{Q}{T}$ for $j<k$\footnote{When $k=1$, we don't need this condition.}, then there exists $m_k\in\mathcal{D}\setminus\mathcal{D}_{k_-}$ such that
\begin{equation}\label{eqn:one_k}
\left|\vtheta_k-\lambda_{m_k}\right|\leq \frac{Q}{T}\,.
\end{equation}

The proof of \eqref{eqn:one_k} consists of two main steps. In the first step, we directly control the error of the filter function to establish a loose bound $T|\lambda_{m_k}-\vtheta_{k}|=\mathcal{O}(\zeta^{-1/2})$.
When $p_{\min}\gg \ptail$, this results in $T|\lambda_{m_k}-\vtheta_{k}|=\mathcal{O}((\ptail/p_{\min})^{1/2})$, a bound that is weaker than the one in \eqref{eqn:one_k}.
This weaker bound allows the confinement of $\vtheta_{m_k}$ to a narrow region around $\lambda_i$. Then we employ a Taylor expansion of the filter function within this restricted region and refine the bound to $T|\lambda_i-\vtheta_{k_i}|=\mathcal{O}(\ptail/p_{\min})$.

Define $m^\star_{k}=\mathrm{argmin}_{m\in\mathcal{D}}\left|\vtheta_k-\lambda_{m}\right|$, $\lambda^\star_k=\lambda_{m^\star_{k}}$, $p^\star_k=p_{m^\star_{k}}$, and $\vtheta^\star_k=\mathrm{argmin}_{\theta_j}\left|\theta_j-\lambda_{m^\star_{k}}\right|$. It is straightforward to see $\lambda^\star_k\notin\mathcal{I}_{B,k}$ because each interval $\left[\vtheta_j-\frac{\alpha}{T},\vtheta_j+\frac{\alpha}{T}\right]$ contains exactly one dominant eigenvalue for every $j\in [k-1]$. Next, because the candidates $\theta_j$ are chosen with the step size $q/T$, we have
\[
\left|\vtheta^\star_k-\lambda^\star_k\right|\leq \frac{q}{T}\,.
\]
Furthermore, because $\alpha/q\in\mathbb{N}$, and $\lambda_{k}^\star$ is not covered by the blocked set, we must have $\vtheta^\star_k\notin \mathcal{I}_{B,k}$. In addition, we have $G(\vtheta^\star_k)$ is close to $G(\lambda^\star_k)$, as given by
\begin{equation}\label{eq:G_theta_star_G_lambda_star}
\begin{aligned}
 &\left|G(\vtheta^\star_k)-G(\lambda^\star_k)\right|\\
\leq &~\left|p^\star_k\left(1-\exp\left(-\frac{T^2\left(\lambda^\star_k-\vtheta^\star_k\right)^2}{2}\right)\right)\right|\\
&+\left|\sum_{m\in \mathcal{D}\setminus\{m^\star_k\}}p_m\exp\left(-\frac{T^2\left(\lambda_m-\lambda^\star_k\right)^2}{2}\right)-p_m\exp\left(-\frac{T^2\left(\lambda_m-\vtheta^\star_k\right)^2}{2}\right)\right|\\
&+\left|G_{\rm tail}(\vtheta^\star_k)-G_{\rm tail}(\lambda^\star_k)\right|+\left|E(\vtheta^\star_k)-E(\lambda^\star_k)\right|,
\end{aligned}
\end{equation}
where $G_{\rm tail}$ is defined in \eqref{eqn:G_tail}. For the first term in \eqref{eq:G_theta_star_G_lambda_star}, since $\left|\vtheta^\star_k-\lambda^\star_k\right|\leq \frac{q}{T}$, we have
\begin{align*}
    \left|p^\star_k\left(1-\exp\left(-\frac{T^2\left(\lambda^\star_k-\vtheta^\star_k\right)^2}{2}\right)\right)\right| = {\cal O}(q^2).
\end{align*}
For the second term in \eqref{eq:G_theta_star_G_lambda_star}, since $T=\Omega(\alpha/\Deltam)$ and $\left|\vtheta^\star_k-\lambda_m\right|\geq \frac{\Deltam}{2}$, we have
\begin{align*}
    \left|\sum_{m\in \mathcal{D}\setminus\{m^\star_k\}}p_m\exp\left(-\frac{T^2\left(\lambda_m-\lambda^\star_k\right)^2}{2}\right)-p_m\exp\left(-\frac{T^2\left(\lambda_m-\vtheta^\star_k\right)^2}{2}\right)\right| = \mathcal{O}(\exp(-\Theta(\alpha^2))).
\end{align*}
For the third term in \eqref{eq:G_theta_star_G_lambda_star}, by the $\ptail T$-Lipschitz property of $G_{\rm tail}(\theta)$ and $\left|\vtheta^\star_k-\lambda^\star_k\right|\leq \frac{q}{T}$, we have
\begin{align*}
    \left|G_{\rm tail}(\vtheta^\star_k)-G_{\rm tail}(\lambda^\star_k)\right| = \mathcal{O}(\ptail q).
\end{align*}
For the last term in \eqref{eq:G_theta_star_G_lambda_star}, by \eqref{eqn:error_small_1_regime_3} and $\left|\vtheta^\star_k-\lambda^\star_k\right|\leq \frac{q}{T}$, we have
\begin{align*}
    \left|E(\vtheta^\star_k)-E(\lambda^\star_k)\right| \leq \zeta\ptail(\sigma q/10+\zeta\ptail/100).
\end{align*}
Combining them, we get that
\begin{align}\label{eqn:difference_small}
    \left|G(\vtheta^\star_k)-G(\lambda^\star_k)\right|\leq \frac{\zeta\ptail^2}{10},
\end{align}
which follows by employing employ \eqref{eqn:condition:para_regime_2} to control the parameters $\sigma,\alpha,q$.

Now we show that $\vtheta_k$ is close to $\lambda^\star_k$ using $\vtheta^\star_k$ as a bridge. First, we notice
\[
\left|\vtheta_k-\lambda_j\right|\geq \frac{\Deltam}{2},\quad \forall j\in{\cal D}\setminus\{m^\star_k\}\,.
\]
Combining this with the condition that $T=\Omega\left(\log^{1/2}(1/(\zeta\ptail))\Deltam^{-1}\right)$, we have
\begin{equation}\label{eqn:other_dominant_ eigenvalue_small_regime_2}
\sum_{m\in \mathcal{D}\setminus\{m^\star_k\}}p_m\exp\left(-\frac{T^2\left(\lambda_m-\lambda^\star_k\right)^2}{2}\right)\leq \frac{\zeta\ptail^2}{10},\quad \sum_{m\in \mathcal{D}\setminus\{m^\star_k\}}p_m\exp\left(-\frac{T^2\left(\lambda_m-\vtheta_{k}\right)^2}{2}\right)\leq \frac{\zeta\ptail^2}{10}\,.
\end{equation}
Next, we use $|G_{\mathrm{tail}}|\leq \ptail$, \eqref{eqn:error_small_1_regime_2}, and \eqref{eqn:other_dominant_ eigenvalue_small_regime_2} to obtain
\begin{equation}\label{eqn:G_square_first}
\begin{aligned}
    G\left(\vtheta_k\right)\leq&~ p^\star_{k}\exp\left(-T^2\left(\lambda^\star_k-\vtheta_{k}\right)^2/2\right)+\ptail+\frac{\zeta\ptail+\zeta\ptail^{2}}{10}\\
    \le&~p^\star_{k}\exp\left(-T^2\left(\lambda^\star_k-\vtheta_{k}\right)^2/2\right)+\ptail+\frac{\zeta\ptail}{5}
\end{aligned}
\end{equation}
and
\begin{equation}\label{eqn:G_square_2_first}
\begin{aligned}
G\left(\lambda^\star_k\right)\geq p^\star_k+\sum_{m\neq \{m^\star_k\}}p_m\exp\left(-\frac{T^2\left(\lambda_m-\lambda^\star_k\right)^2}{2}\right)-\left|E(\lambda^\star_k)\right|\geq p^\star_k-\frac{\zeta\ptail}{10}\,.
\end{aligned}
\end{equation}
Furthermore, \eqref{eqn:difference_small} and the assumption that $\vtheta_k$ is the maximal point imply that
\begin{equation}\label{eqn:G_inequal_first}
G\left(\vtheta_k\right)\geq G\left(\vtheta^\star_k\right)\geq G\left(\lambda^\star_k\right)-\frac{\zeta\ptail}{10}\,.
\end{equation}
Combining \eqref{eqn:G_square_first}, \eqref{eqn:G_square_2_first}, and \eqref{eqn:G_inequal_first}, we obtain
\[
p^\star_{k}\left(1-\exp\left(-T^2\left(\lambda^\star_k-\vtheta_{k}\right)^2/2\right)\right)\leq \left(1+2\zeta/5\right)\ptail\,,
\]
which implies
\begin{equation}\label{eqn:rough_bound_first}
\left|T(\lambda^\star_k-\vtheta_k)\right|\leq \sqrt{-2\ln\left(1-\frac{\left(1+2\zeta/5\right)\ptail}{p_{\min}}\right)}=\mathcal{O}\left(\zeta^{-1/2}\right)\,.
\end{equation}

We note that, in the case where $\zeta^{-1}=\mathcal{O}\left(\ptail/p_{\min}\right)$, the upper bound proved above scales with $\left(\ptail/p_{\min}\right)^{1/2}$, while the goal is to improve the upper bound to $Q$, which linearly depends on $\ptail/p_{\min}$. Now we show how to improve \eqref{eqn:rough_bound_first}.
Similar to \eqref{eqn:G_square_first} and \eqref{eqn:G_square_2_first},
\begin{equation}\label{eqn:G_square}
\begin{aligned}
&\left|G^2\left(\vtheta_k\right)-(p^\star_{k})^2\exp\left(-T^2\left(\lambda^\star_k-\vtheta_{k}\right)^2\right)\right.\\
&\left.-2p^\star_{k}\mathrm{Re}\left(\exp\left(-T^2\left(\lambda^\star_k-\vtheta_{k}\right)^2/2\right)\left(G_{\mathrm{tail}}(\vtheta_k)+E(\vtheta_k)\right)\right)\right|=\mathcal{O}(\ptail^{2})
\end{aligned}
\end{equation}
and
\begin{equation}\label{eqn:G_square_2}
\left|G^2\left(\lambda^\star_k\right)-\left[(p^\star_{k})^2+2p^\star_{k}\mathrm{Re}\left(\left(G_{\mathrm{tail}}(\lambda^\star_k)+E(\lambda^\star_k)\right)\right)\right]\right|=\mathcal{O}(\ptail^{2})\,.
\end{equation}
Furthermore, \eqref{eqn:difference_small} and the assumption that $\vtheta_k$ is the maximal point imply that
\begin{equation}\label{eqn:G_inequal}
G^2\left(\vtheta_k\right)\geq G^2\left(\vtheta^\star_k\right)\geq G^2\left(\lambda^\star_k\right)-\mathcal{O}(\ptail^{2})\,.
\end{equation}
Combining \eqref{eqn:G_square}, \eqref{eqn:G_square_2}, and \eqref{eqn:G_inequal}, we obtain
\[
\begin{aligned}
&p^\star_{k}\left(\exp\left(-T^2\left(\lambda^\star_k-\vtheta_{k}\right)^2\right)-1\right)\\
\geq&~2\mathrm{Re}\left(-\exp\left(-T^2\left(\lambda^\star_k-\vtheta_{k}\right)^2/2\right)\left(G_{\mathrm{tail}}(\vtheta_k)+E(\vtheta_k)\right)+\left(G_{\mathrm{tail}}(\lambda^\star_k)+E(\lambda^\star_k)\right)\right)-
O(\ptail^{2}/p_{\min})\\
\geq &~2\mathrm{Re}\left(\left(1-\exp\left(-T^2\left(\lambda^\star_k-\vtheta_{k}\right)^2/2\right)\right)\left(G_{\mathrm{tail}}(\vtheta_k)+E(\vtheta_k)\right)\right)\\
&+2\mathrm{Re}\left(-\left(G_{\mathrm{tail}}(\vtheta_k)+E(\vtheta_k)\right)+\left(G_{\mathrm{tail}}(\lambda^\star_k)+E(\lambda^\star_k)\right)\right)-
O(\ptail^{2}/p_{\min})\\
\geq &~2\mathrm{Re}\left(\left(1-\exp\left(-T^2\left(\lambda^\star_k-\vtheta_{k}\right)^2/2\right)\right)E(\vtheta_k)\right)\\
&+2\mathrm{Re}\left(-\left(G_{\mathrm{tail}}(\vtheta_k)+E(\vtheta_k)\right)+\left(G_{\mathrm{tail}}(\lambda^\star_k)+E(\lambda^\star_k)\right)\right)-
O(\ptail^{2}/p_{\min})\,.
\end{aligned}
\]

Define $F_k(\theta)=\exp\left(-\frac{T^2\left(\lambda^\star_k-\theta\right)^2}{2}\right)$\footnote{
We note that $F_k(\theta)\approx \int^\infty_{-\infty}a(t)\exp(i(\lambda^\star_k-\theta)t)dt=F(\lambda^\star_k-\theta)$, where $F$ is defined as the Fourier transform of $a(t)$ (refer to \eqref{eqn:fourier_transform_a}).
}.
Then, the above inequality can be rewritten as
\[
\begin{aligned}
&p^\star_k(1-F_k(\vtheta_k))(1+F_k(\vtheta_k))+2\mathrm{Re}\left((1-F_k(\vtheta_k))E(\vtheta_k)\right)\\
\leq &2\left|G_{\mathrm{tail}}(\vtheta_k)-G_{\mathrm{tail}}(\lambda^\star_k)\right|+2\left|E(\vtheta_k)-E(\lambda^\star_k)\right|+\mathcal{O}(\ptail^{2}/p_{\min})\,.
\end{aligned}
\]
Because $|E(\vtheta_k)|\leq \zeta\ptail/10$ according to \eqref{eqn:error_small_1_regime_2}, the above inequality further implies
\begin{equation}\label{eqn:important_inequa_regime_2}
\begin{aligned}
    &\left(p_{\min}-\zeta\ptail/5\right)\left(1-F_k(\vtheta_k)\right)\\
    \leq &2\left|G_{\mathrm{tail}}(\vtheta_k)-G_{\mathrm{tail}}(\lambda^\star_k)\right|+2\left|E(\vtheta_k)-E(\lambda^\star_k)\right|+\mathcal{O}(\ptail^{2}/p_{\min})\\
    \leq &2(\sigma\zeta+1)T\ptail\left|\vtheta_k-\lambda^\star_k\right|+\mathcal{O}\left(\ptail^{2}/p_{\min}\right)\,.
\end{aligned}
\end{equation}
where we use the fact that $F_k(\vtheta_k)\in[0,1]$ in the first inequality, and use \eqref{eqn:error_small_1_regime_3} and the fact that $G_{\mathrm{tail}}$ is $\ptail T$-Lipschitz in the last inequality.

Using \eqref{eqn:rough_bound_first}, we obtain
\[
1-F_k(\vtheta_k)= 1-\exp\left(-\frac{T^2\left(\vtheta_k-\lambda^\star_k\right)^2}{2}\right)\geq \frac{\exp\left(-\Theta(\zeta^{-1})\right)T^2\left(\vtheta_k-\lambda^\star_k\right)^2}{2}\,.
\]
Plugging this into \eqref{eqn:important_inequa_regime_2}, we obtain
\[
\left(p_{\min}-\zeta\ptail/5\right)\frac{\exp\left(-\Theta(\zeta^{-1})\right)}{2}T^2\left(\vtheta_k-\lambda^\star_k\right)^2-2(\sigma\zeta+1)\ptail T\left|\vtheta_k-\lambda^\star_k\right|=\mathcal{O}\left(\ptail^{2}/p_{\min}\right)\,.
\]
Viewing the left-hand side as a quadratic function with respect to $T\left|\vtheta_k-\lambda^\star_k\right|$, we obtain
\begin{equation}\label{eqn:final_inequa_regime_2}
\left|\vtheta_k-\lambda^\star_k\right|=\mathcal{O}\left(\frac{\exp\left(\Theta(\zeta^{-1})\right)(\sigma\zeta+1)\ptail}{\left(p_{\min}-\zeta\ptail/5\right)T}\right)\leq \frac{Q}{T}\,.
\end{equation}
This concludes the proof.
\end{proof}

We prove \cref{thm:regime_3} as follows.
\begin{proof}[Proof of \cref{thm:regime_3}]
The proof is similar to the proof of \cref{thm:regime_2}. For simplicity, we omit some details in this proof.

By \cref{lem:truncation_error} with $\delta=\zeta$ and $J=\mathcal{O}(T/q)$, it holds that
\begin{equation}\label{eqn:random_error_regime3}
\mathbb{P}\left(\max_{\theta\in\Theta}\left|E(\theta)\right|=\mathcal{O}(\zeta)\right)\geq 1-\eta\,,
\end{equation}
and
\begin{equation}\label{eqn:random_error_2_regime3}
\mathbb{P}\left(\cap_{\theta,\theta'\in\Theta,\theta\neq \theta'}\left\{\left|E(\theta)-E(\theta')\right|=\mathcal{O}\left(\sigma T\zeta|\theta-\theta'|+\zeta\right)\right\}\right)\geq 1-\eta\,,
\end{equation}
where $\Theta=\left\{\theta_j\right\}^J_{j=1}\cup \left\{\lambda_m\right\}_{m\in\mathcal{D}}$. Thus, it suffices to prove \eqref{eqn:distance_bound_regime_3} assuming \eqref{eqn:random_error_regime3} and \eqref{eqn:random_error_2_regime3} hold.

Because $T=\Omega\left(\frac{\alpha}{\Delta}\right)$, for any $j\in [k-1]$, the interval $\left[\vtheta_j-\frac{\alpha}{T},\vtheta_j+\frac{\alpha}{T}\right]$ contains exactly one dominant eigenvalue, and those intervals that contain $\{\lambda_j\}_{j\in {\cal D}}$ will not contain any $\lambda_{j'}$ for $j'\in {\cal D}^c$. Define $m^\star_{k}:=\mathrm{argmin}_{m\in\mathcal{D}}\left|\vtheta_k-\lambda_{m}\right|$, $\lambda^\star_k:=\lambda_{m^\star_{k}}$, $p^\star_k:=p_{m^\star_{k}}$, and $\vtheta^\star_k:=\mathrm{argmin}_{\theta_j}\left|\theta_j-\lambda_{m^\star_{k}}\right|$. Similar to the argument as the proof of \cref{thm:regime_2}, we know that $\left|\vtheta^\star_k-\lambda^\star_k\right|\leq \frac{q}{T}$ and $\vtheta^\star_k\notin\mathcal{I}_{B,k}$.
In the following proof, we show that $|\vtheta_k-\lambda_k^\star|\leq \frac{Q}{T}$. Define the filter function that contains the impact of other dominant eigenvalues:
\begin{align*}
    G_{\rm dom}(\theta):=\sum_{m\in {\cal D}\backslash\{m_k^\star\}} p_m\exp\left(-\frac{T^2\left(\lambda_m-\theta\right)^2}{2}\right).
\end{align*}
We first prove that $G_{\rm dom}(\lambda_k^\star)$, $G_{\rm dom}(\vtheta_k^\star)$, and $G_{\rm dom}(\vtheta_k)$ are of order $O(\zeta^2)$ similar to the proof of \cref{thm:regime_2}. Because $\lambda^\star_k$ is the closest dominant eigenvalue for $\vtheta^\star_k$ and $\vtheta_k$ and $\alpha/T<\Delta/2$, we obtain
\[
|\lambda_k^\star-\lambda_j|\geq \Delta,\quad \left|\vtheta_k-\lambda_j\right|\geq \frac{\Delta}{2},\quad \left|\vtheta_k-\lambda_j\right|\geq \frac{\Delta}{2},\quad \forall j\in{\cal D}\setminus\{m^\star_k\}\,.
\]
Combining this with $T=\Omega(\log^{1/2}(1/\zeta)/\Delta)$, we obtain
\begin{align}
    &\max\left\{G_{\rm dom}(\lambda_k^\star),G_{\rm dom}(\vtheta_k^\star),G_{\rm dom}(\vtheta_k)\right\}
    \leq \sum_{m\in {\cal D}\backslash \{m_k^\star\}} p_m\exp\left(-\Omega(T^2\Delta^2)\right)
    =  O(\zeta^2)\,.\label{eq:bound_dom_lambda_star_regime}
\end{align}
Then, we upper bound the tail errors $G_{\rm tail}(\lambda_k^\star)$ and $G_{\rm tail}(\vtheta_k)$. For $G_{\rm tail}(\lambda_k^\star)$, since $|\lambda_k^\star-\lambda_m|\geq \Delta$ for any $m\in [M]$, using the same calculation as \eqref{eq:bound_dom_lambda_star_regime}, it holds that
\begin{align}\label{eq:bound_tail_lambda_star_regime3}
    G_{\rm tail}(\lambda_k^\star)=O(\zeta^2\ptail).
\end{align}
where we also use $\sum_{m\in\mathcal{D}^c}p_m\leq \ptail$. For $G_{\rm tail}(\vtheta_k)$, according to \cref{thm:regime_1}, we have $\left|\vtheta_k-\lambda^\star_k\right|\leq \alpha/T\leq \Delta/2$. This implies $|\vtheta_k-\lambda_m|>\frac{\Delta}{2}$ for any $m\in {\cal D}^c$. Therefore, using a similar calculation as \eqref{eq:bound_dom_lambda_star_regime}, we can show that
\begin{align}\label{eq:bound_tail_vtheta_regime3}
    G_{\rm tail}(\vtheta_k)=O(\zeta^2\ptail).
\end{align}
Now, we are ready to upper bound $|\vtheta_k-\lambda_k^\star|$. Similar to \eqref{eqn:difference_small}, we can show $|G(\vtheta_k^\star) - G(\lambda_k^\star)| = O(\zeta^2)$. This implies
\begin{align}\label{eqn:order_relation_regime_3}
    G(\vtheta_k)^2 \geq G(\vtheta_k^\star)^2\geq G(\lambda_k^\star)^2 -O(\zeta^2)\,.
\end{align}
Define $F_k(\theta)=\exp\left(-\frac{T^2\left(\lambda^\star_k-\theta\right)^2}{2}\right)$. Since $G(\theta)=|p_k^\star F_k(\theta) + E(\theta) + G_{\rm dom}(\theta)+G_{\rm tail}(\theta)|$,
we have
\begin{align*}
    G(\vtheta_k)^2
    =&~ {p_k^\star}^2F_k(\vtheta_k)^2 + |E(\vtheta_k)|^2 + G_{\rm dom}(\vtheta_k)^2 + G_{\rm tail}(\vtheta_k)^2\\
    &+2p_k^{\star} F_k(\vtheta_k)({\rm Re}(E(\vtheta_k))+G_{\rm dom}(\vtheta_k) + G_{\rm tail}(\vtheta_k)) \\
    &+ 2({\rm Re}(E(\vtheta)) (G_{\rm dom}(\vtheta_k) + G_{\rm tail}(\vtheta_k)) + 2G_{\rm dom}(\vtheta_k) G_{\rm tail}(\vtheta_k)\\
    = &~ {p_k^\star}^2F_k(\vtheta_k)^2 + G_{\rm tail}(\vtheta_k)^2+ 2p_k^{\star} F_k(\vtheta_k)({\rm Re}(E(\vtheta_k)) + G_{\rm tail}(\vtheta_k))+2{\rm Re}(E(\vtheta_k))G_{\rm tail}(\vtheta_k)+O(\zeta^2)\\
    = &~ {p_k^\star}^2F_k(\vtheta_k)^2 + 2p_k^{\star} F_k(\vtheta_k){\rm Re}(E(\vtheta_k))  + O(\zeta^2)\,,
\end{align*}
where the second step follows from \eqref{eqn:random_error_regime3} and \eqref{eq:bound_dom_lambda_star_regime}, and the third step follows from \eqref{eq:bound_tail_vtheta_regime3}. Similarly, for $G(\lambda_k^\star)^2$, by \eqref{eqn:random_error_regime3}, \eqref{eq:bound_dom_lambda_star_regime}, and \eqref{eq:bound_tail_lambda_star_regime3}, we have the following:
\begin{align*}
    G(\lambda_k^\star)^2\geq {p_k^\star}^2 +  2p_k^{\star} {\rm Re}(E(\lambda_k^\star))-O(\zeta^2)
\end{align*}
Similar to the \eqref{eqn:important_inequa_regime_2}, the above two inequalities, \eqref{eqn:random_error_2_regime3}, and \eqref{eqn:order_relation_regime_3} imply
\begin{align}\label{eq:final_ineq_regime3}
    (p_{\min}-2\zeta)(1-F_k(\vtheta_k)) \leq O(\zeta)\cdot \min\{\sigma T |\vtheta_k - \lambda_k^\star|, 1\} +O(\zeta^2/p_{\min})\,.
\end{align}
Similar to \eqref{eqn:important_inequa_regime_2}-\eqref{eqn:final_inequa_regime_2}, using $\zeta<p_{\min}/4$, we further obtain
\begin{align*}
    |\vtheta_k-\lambda_k^\star|=O\left(\frac{\zeta\sigma}{p_{\min}-2\zeta} \frac{1}{T}\right)\leq \frac{Q}{T}\,.
\end{align*}
The theorem is then proved.
\end{proof}

\section{Other quantum phase estimation algorithms}\label{sec:summary_numerical_method}
In this section, we give a brief summary of the previous quantum phase estimation algorithms that are tested in \cref{sec:num}: MM-QCELS~\cite{Ding2023simultaneous}, QPE (textbook version~\cite{NielsenChuang2000}), and ESPRIT~\cite{Stroeks_2022}.

\begin{itemize}
\item (MM-QCELS~\cite{Ding2023simultaneous}): The dataset used in MM-QCELS is similar to QMEGS (\cref{alg:main}) and is also generated by \cref{alg:data}. The main subroutine of MM-QCELS is called quantum complex exponential least squares (QCELS): Given a data set $\left\{(t_{n},Z_{n})\right\}^{N}_{n=1}$ generated from Algorithm \ref{alg:data}, MM-QCELS obtains an estimate for the dominant eigenvalues by solving the following optimization problem:
\begin{equation}\label{eqn:op}
\left(\{r^*_k\}^{K}_{k=1},\{\theta^*_k\}^{K}_{k=1}\right)=\argmin_{r_k\in\mathbb{C},\theta_k\in\mathbb{R}}L_{K}\left(\{r_k\}^{K}_{k=1},\{\theta_k\}^{K}_{k=1}\right)\,.
\end{equation}
with loss function
\begin{equation}\label{eqn:loss_multi_modal}
L_{K}\left(\{r_k\}^{K}_{k=1},\{\theta_k\}^{K}_{k=1}\right)=\frac{1}{N}\sum^N_{n=1}\left|Z_n-\sum^{K}_{k=1}r_k\exp(-i\theta_k t_n)\right|^2\,.
\end{equation}
Choosing $K=|\mathcal{D}|$ and a proper generated data set $\left\{(t_{n},Z_{n})\right\}^{N}_{n=1}$\footnote{In~\cite[Algorithm 2]{Ding2023simultaneous}, the authors need to generate a sequence of data set $\left\{(t_{n},Z_{n})\right\}^{N}_{n=1}$ with different $T$ and $N$ to ensure Heisenberg limit scaling of the algorithm theoretically.}, \cite{Ding2023simultaneous} shows that the solution $\{\theta^*_k\}^K_{k=1}$ of~\eqref{eqn:op} is a good approximation to the set of dominant eigenvalues $\{\lambda_m\}_{m\in\mathcal{D}}$. MM-QCELS can reach Heisenberg limit scaling and small circuit depth when $\ptail\ll 1$. However, the algorithm requires a spectral gap assumption, meaning $T=\Omega(1/\epsilon)>\Deltam^{-1}$, to ensure that the optimization problem can differentiate between the dominant eigenvalues.

\item (QPE (textbook version~\cite{NielsenChuang2000})): {We provide a brief review of QPE for ground state energy estimation.} The quantum process involves a sequence of controlled time evolution operations $e^{-iH}$ on a state $\ket{0^d}|\psi\rangle=\sum^{M-1}_{m=0} c_{m}\ket{0^d}\left|\psi_{m}\right\rangle$. Here, $\left|\psi_{m}\right\rangle$ represents the eigenstates associated with eigenvalues $\lambda_{m}$. The resulting quantum state from these operations, prior to applying the inverse Quantum Fourier Transform (QFT), is expressed as follows:
\[
|\Psi\rangle=\frac{1}{\sqrt{N_t}} \sum_{j=-N_t / 2}^{N_t / 2-1}|j\rangle e^{-i j H}|\psi\rangle\,,
\]
where $N_t=2^d$. After applying the inverse QFT and measuring the ancilla register, the probability of obtaining outcome $k$ is given by:
\begin{equation}\label{eqn:PK}
P(k)=\sum^{M-1}_{m=0}\left|c_{m}\right|^{2} K_{N_t}\left(\frac{2 \pi k}{N_t}- \lambda_{m}\right),
\end{equation}
where $-N_t/2\leq k\leq N_t/2-1$, and $K_{N_t}$ is the squared and normalized Dirichlet kernel defined as $K_{N_t}(\theta)=\frac{\sin ^{2}(\theta N_t / 2)}{N_t^{2} \sin ^{2}(\theta / 2)}$. To simulate QPE classically, we sample this distribution $N_{QPE}$ times to obtain a set of samples $\{k_i\}^{N_{QPE}}_{i=1}$. The ground state energy can then be approximated as $\widetilde{\lambda}_0=\frac{2\pi\min_{i}k_i}{N_t}$.

We note that the original textbook version of QPE algorithm in \cite[Chapter 5.2]{NielsenChuang2000} is designed to estimate general eigenvalues, not necessarily the ground state energy. The version presented above is a variation specifically designed for ground state energy estimation. There are also other variations of QPE, such as the Gaussian/Kaiser window-based QPE \cite{Berry_2024, mcardle2022quantumstate}, where the authors implement different resource states to produce a more concentrated kernel and reduce $T_{\rm total}$. In addition, it might be possible to extend QPE to estimate multiple eigenvalues simultaneously: For example, in the ideal case where $p_{\text{tail}} = 0$ and $ \{2\pi\lambda_m\}_{m\in \mathcal{D}} $ are finite-digit numbers. In this ideal scenario, when $ N_t $ is sufficiently large, we have
\[
P(k) = \sum_{m\in\mathcal{D}} \left| c_{m} \right|^2 \delta_0\left( \frac{2\pi k}{N_t} - \lambda_m \right)\,.
\]
Therefore, measuring the ancilla qubits $\mathcal{O}(\mathrm{poly}(K))$ times is sufficient to obtain all the eigenvalues exactly.
In the non-ideal case, we need to design an efficient method to post-process the output of QPE and obtain accurate estimations of multiple eigenvalues. To the best of our knowledge, we are unaware of any such procedure that has clear complexity analysis. Thus, for simplicity, we mainly consider the simple version of QPE in our paper.

\item (ESPRIT~\cite{Stroeks_2022}): The dataset used in ESPRIT is similar to \cref{alg:main} and is also generated by \cref{alg:data}. The algorithm of ESPRIT relies on the construction and manipulation of the Hankel matrix: Given $T>0$ and an odd integer $N>0$, we first set $t_n=n\tau$ for $0\leq n\leq N$, where $\tau=T/N$, and construct Hankel matrix $\mathrm{H}\in\mathbb{C}^{\frac{N+1}{2}\times \frac{N+1}{2}}$ with $\mathrm{H}_{i,j}=Z_{t_{i+j}}$. Here, $Z_{t_n}$ is generated by \cref{alg:data}. We then find the singular value decomposition $\mathrm{H}=U\Sigma V^\dagger$ and define
\[
U_0=U[:-1,:K],\quad U_1=[1:,:K]\,.
\]
Here, $U_0$ contains first $\frac{N-1}{2}$ rows and first $K$ columns of $U$ and $U_1$ contains the last $\frac{N-1}{2}$ rows and first $K$ columns of $U$. Finally, we find the eigenvalues $\{\mu_k\}^K_{k=1}$ of $U^{-1}_1U_0$ ($U^{-1}_1$ is the pseudoinverse of $U_1$) and define $\{\vtheta_k=\mathrm{angle}(\mu_k)/\tau\}^K_{k=1}$. According to the results of classical signal processing~\cite{9000636}, we can demonstrate that, when $K=|\mathcal{D}|$ and $T,N$ are chosen properly, $\{\vtheta_k=\mathrm{angle}(\mu_k)/\tau\}^K_{k=1}$ is a set that is close to the set of dominant eigenvalues $\{\lambda_m\}_{m\in\mathcal{D}}$.

We would like to highlight that the original ESPRIT method falls short of achieving the Heisenberg limit scaling in the context of quantum phase estimation. In the original ESPRIT framework, the choice of $\tau=1$ is imperative to mitigate aliasing issues\footnote{It is not possible to differentiate between $\lambda_k$ and $\lambda_k+2\pi\tau$ in ESPRIT, as they produce the same data.}. Consequently, this requires $N=T$, $T_{\max}=T$, and results in $T_{\mathrm{total}}=\Theta(T^2_{\max})$. However, according to the generalized uncertainty relation~\cite{Braunstein1996GeneralizedUR}, there exists a uniform complexity lower bound for phase estimation~\cite{PhysRevLett.96.010401}, asserting that the square of the error is at least $\Omega\left(T_{\mathrm{total}}^{-1}T_{\max}^{-1}\right)$ in expectation. By combining $T_{\mathrm{total}}^{-1}T_{\max}^{-1}=\mathcal{O}(\epsilon^2)$ and $T_{\mathrm{total}}=\Theta(T^2_{\max})$, we deduce $T_{\mathrm{total}}=\Omega(\epsilon^{-4/3})$. The crux of the matter, as deduced from the previous analysis, is that the necessary selection of $\tau=1$ to avoid aliasing poses a significant hurdle to ESPRIT in achieving the Heisenberg limit.

More recently,~\cite{ni2023lowdepth_2} proposes a multilevel ESPRIT approach to circumvent aliasing issues without enforcing $\tau=1$. Notably, ~\cite{ni2023lowdepth_2} generates a sequence of datasets with carefully chosen values for $T$ and $N$, progressively refining the estimation of dominant eigenvalues. The successful application of multilevel techniques enables them to achieve Heisenberg-limited scaling and shorter circuit depth. In our numerical simulations presented in \cref{sec:num}, for simplicity, we only consider the original ESPRIT and choose $N=T$ and $\tau=1$ to mitigate aliasing issues associated with ESPRIT.

\end{itemize}

\section{Algorithm for integer-power setting}\label{sec:intpower}
In certain phase estimation tasks, only a black box unitary $U$ represented by a quantum circuit can be accessed.
Under this setting, querying an arbitrary real power of $U$ is not feasible. Instead, only integer powers of it can be acquired. To maintain consistency with the notation used in the case where $U=e^{-iH}$, we still assume $U\ket{\psi_m} = e^{-i\la_m}\ket{\psi_m}$. In this case, we aim to recover the phases $\lambda_m \mod 2\pi$. All other parameters, such as $p_{\min}$, $\ptail$, and $\mathcal{D}$, retain the same definitions as in the real-power setting unless explicitly stated otherwise. We define the$\mod 2\pi$ distance of two numbers $u$ and $v$ as
$$|u-v|_{2\pi}:= \min\{|u-v \mod 2\pi|,\, |v-u \mod 2\pi|\}.$$

The idea is similar to the real power setting, which is to leverage the Gaussian-filtered spectral density. The difference is that we will use the periodic Gaussian
\begin{equation}
    \phi_p(x) =W\sum_{j\in\ZZ} e^{- \frac{(x+2j\pi)^2T^2}{2}},
\end{equation}
where $W<1$ is a normalizing constant such that $\phi_p(0) = 1$. It is clear that $\phi_p$ is $2\pi$-periodic, and its Fourier coefficients are
\begin{equation}
    \hat{\phi}_p(k) = \frac{1}{2\pi}\int_0^{2\pi}\phi_p(x)e^{-ikx}\d x = \frac{W}{\sqrt{2\pi}T} e^{-\frac{k^2}{2T^2}}.
\end{equation}
Therefore, we have $\sum_{k\in\ZZ}\hat{\phi}_p(k) = \phi_p(0) = 1$. Let
$$a(k) = \begin{cases}
    \hat{\phi}_p(0)+\sum_{|j|>\sigma T}\hat{\phi}_p(j)& k=0,\\
    \hat{\phi}_p(k)& 1\le|k|\le \sigma T,\\
    0&|k|>\sigma T.
\end{cases}$$
be a distribution over $\ZZ$, and $t$ be a random variable sampled from this distribution. $Z_t$ denote the unbiased estimation of $\braket{\psi|U^{t}|\psi}$ obtained by the Hadamard test. Therefore, the maximal quantum runtime $T_{\max}$ is bounded by $\sigma T$, as similar to the real-power setting. Except for the distribution of $t$, the rest part of the phase estimation algorithm goes the same as in \Cref{alg:main}.  We may also establish theorems that guarantee the performance of this slightly modified algorithm. First, we establish a lemma that gives several properties of the periodic Gaussian function $\phi_p(x)$.

\begin{lemma}\label{lemma:periodic guassian bound}
    If $T\ge 1$, then $\phi_p(x)$ is increasing on $[-\pi,0]$ and decreasing on $[0,\pi]$. For $x\in [-\frac{2\pi}{3},\frac{2\pi}{3}]$,
    \begin{equation}
        e^{-\frac{x^2T^2}{2}} \le \phi_p(x) \le 1.01 e^{-\frac{x^2T^2}{2}},\label{ineq:gauss bound}.
    \end{equation}
\end{lemma}
\begin{proof}
    The monotonicity part is proved in \cite[Lemma 2]{li2023101577}, where only some normalizing constants differ. Without loss of generality, we will prove the rest of the lemma assuming $x\ge 0$ since $\phi_p(x)$ is an even function. For the left part of \eqref{ineq:gauss bound}, we need to notice that when $x=0$, the equality holds. Moreover, when $x\in[0,\pi]$, we have
    \[
    \begin{aligned}
        \frac{\d}{\d x}\phi_p(x) &= WT^2\left(-xe^{-\frac{T^2x^2}{2}}+\sum_{j=1}^{+\infty}\left(-(2j\pi+x)e^{-\frac{T^2(2j\pi+x)^2}{2}} + (2j\pi-x)e^{-\frac{T^2(2j\pi-x)^2}{2}}\right)\right)\\
        &\ge WT^2\left(-xe^{-\frac{T^2x^2}{2}}\right)\ge -T^2xe^{-\frac{T^2x^2}{2}} = \frac{\d}{\d x}e^{-\frac{T^2x^2}{2}},
    \end{aligned}
    \]
    where in the first inequality, we used the fact that each term of the summation is positive, and in the second inequality, we used $W\le 1$. To see this, we may introduce the function $h(y):= ye^{-\frac{T^2y^2}{2}}$, which is decreasing when $y\ge 1 \ge\frac{1}{T}$, and thus $h(2j\pi-x)\ge h(2j\pi+x)$ for all $j\ge 1$.

    Next, we prove the right part of \eqref{ineq:gauss bound}. This can be done by the following calculation.
    \[
    \begin{aligned}
        \phi_p(x)e^{\frac{x^2T^2}{2}} &= W\sum_{j\in\ZZ}e^{-2T^2\pi(\pi j^2+xj)} \le \sum_{j\in\ZZ}e^{-2T^2\pi(\pi j^2+xj)}\\
        & = 1+e^{-2T^2\pi(\pi-x)}+\sum_{j = 1}^{+\infty}e^{-2T^2\pi(\pi j^2+xj)} + \sum_{j = 2}^{+\infty}e^{-2T^2\pi(\pi j^2-xj)}\\
        &\le 1+e^{-2\pi(\pi/3)}+\sum_{j = 1}^{+\infty}e^{-2\pi j} + \sum_{j = 2}^{+\infty}e^{-2\pi j} < 1.01.
    \end{aligned}
    \]
\end{proof}

Now, we may define the error function
$$E\left(\theta\right)=\frac{1}{N}\sum^N_{n=1}Z_n\exp(i\theta t_n)-\sum^M_{m=1}p_m\phi_p(\theta-\la_m)$$
and $E_j := E(\theta_j)$. Then we have the following lemma similar to \Cref{lem:truncation_error}.

\begin{lemma}\label{lem:int_truncation_error} Given $\delta>0$ and the pair of overlap-eigenvalues $\left\{\left(p_i,\lambda_i\right)\right\}^M_{m=1}$, if $\sigma=\Omega\left(\log^{1/2}\left(1/\delta\right)\right)$, we have
	\begin{equation}\label{eqn:int_time_truncation_error}
		\left|\mathbb{E}\left(\frac{1}{N}\sum^N_{n=1}Z_n\exp(i\theta_j t_n)\right)-\sum^M_{m=1}p_m\phi_p(\theta-\la_m) \right|\leq \delta^2\,.
	\end{equation}
	Furthermore, given $\eta>0$ and $\Theta=\left\{\theta_j\right\}^J_{j=1}\cup \left\{\lambda_m\right\}_{m\in\mathcal{D}}$, if $N=\Omega\left(\frac{1}{\delta^2}\log\left(\left(\frac{T}{q}+|\mathcal{D}|\right)\frac{1}{\eta}\right)\right)$, we have
	\begin{equation}\label{eqn:int_random_error}
		\mathbb{P}\left(\max_{\theta\in\Theta}\left|E(\theta)\right|\leq \delta\right)\geq 1-\eta\,.
	\end{equation}
\end{lemma}
\begin{proof} According to Algorithm \ref{alg:data}, when $\sigma=\Omega\left(\log^{1/2}\left(1/\delta\right)\right)$, we have
	\begin{equation}
		\begin{aligned}
			&\left|\mathbb{E}\left(\frac{1}{N}\sum^N_{n=1}Z_n\exp(i\theta_j t_n)\right)-\sum^M_{m=1}p_m\phi_p(\theta-\la_m)\right|\\
			=&\left|\sum_{m=1}^M p_m\sum_{|k|>\sigma T}\hat{\phi}_p(k)e^{ik(\theta-\la_m)} \right|
			\le  \sum_{|k|>\sigma T}\hat{\phi}_p(k)e^{-\frac{t^2}{2T^2}}\left|\sum^M_{m=1}p_me^{ik(\theta-\la_m)}\right|\le\sum_{|k|>\sigma T}\hat{\phi}_p(k)\\
			=&\sum_{|k|>\sigma T}\frac{W}{\sqrt{2\pi}T} e^{-\frac{k^2}{2T^2}} \le 2\int_{\sigma T}^\infty\frac{1}{\sqrt{2\pi} T}e^{-\frac{t^2}{2T^2}}\d t=\sqrt{\frac{2}{\pi}}\int_{\sigma}^\infty e^{-\frac{s^2}{2}}\d s< e^{-\sigma^2}<\delta^2,		\end{aligned}
	\end{equation}
	where we used the fact that $W<1$ and bounded the summation using integration. This proves \eqref{eqn:int_time_truncation_error}. The proof of \eqref{eqn:int_random_error} is the same as the proof of \cite[Appendix B.3 Lemma 4 eqn. (B8)]{ding2023robust}, thus, we omit it here.
\end{proof}
Define the magnitude function:
\[
G(\theta) = \left|\frac{1}{N}\sum^N_{n=1}Z_n\exp(i\theta t_n)\right|\,.
\]
and $G_j:= G(\theta_j)$.

Finally, we have a similar theorem as \Cref{thm:regime_1}. This shows this algorithm restricted to integer powers of $U$ can also achieve the Heisenberg limit without any gap assumptions.

\begin{theorem}[$\forall T\ge 1$]\label{thm:int_regime_1} Assume $p_{\min}>\ptail$ and $|\mathcal{D}|\leq K$. Given the probability of failure $\eta>0$, we choose the following parameters:
\begin{itemize}
    \item Block constant: $\alpha=\Omega\left(\log^{1/2}\left(\frac{1}{p_{\min}-\ptail}\right)\right)$,
    \item Searching parameter: $q=\mathcal{O}\left(\log^{1/2}\left(\frac{p_{\min}}{\ptail+(p_{\min}-\ptail)/2}\right)\right)$, $q<\alpha$, and $\alpha/q\in\mathbb{N}$,
    \item Time truncation parameter: $\sigma=\Omega\left(\log^{1/2}\left(\frac{1}{p_{\min}-\ptail}\right)\right)$,
    \item Number of samples: $N=\Omega\left(\frac{1}{(p_{\min}-\ptail)^2}\log\left(\left(\frac{T}{q}+|\mathcal{D}|\right)\frac{1}{\eta}\right)\right)$.
\end{itemize}
Then, with probability at least $1-\eta$, we have that for each $i\in\mathcal{D}$, there exists $1\leq k_i\leq |\mathcal{D}|$ such that
\begin{equation}\label{eqn:int_distance_bound_regime_1}
\left|\lambda_i-\vtheta_{k_i}\right|_{2\pi}\leq \frac{\alpha}{T}\,.
\end{equation}
In particular, for any $\epsilon>0$, to achieve
\[
\left\{\lambda_m\right\}_{m\in\mathcal{D}}\subset \cup_{k}[\vtheta_k-\epsilon,\vtheta_k+\epsilon]\mod 2\pi,
\]
it suffices to choose
\[
T_{\max}=\widetilde{\Theta}\left(\frac{1}{\epsilon}\right),\quad T_{\mathrm{total}}=\widetilde{\Theta}\left(\frac{1}{(p_{\min}-\ptail)^2\epsilon}\right)\,,
\]
where the logarithmic factor is omitted.
\end{theorem}

\begin{proof}
Similar to the proof of \Cref{thm:regime_1}, we only need to prove
\begin{equation}\label{eqn:int_separation_regime_1}
\max_{\theta_j\notin \cup_{m\in\mathcal{D}}\left[\lambda_m-\frac{\alpha}{3T},\lambda_m+\frac{\alpha}{3T}\right]}G_j<\min_{\theta_j\in \cup_{m\in\mathcal{D}}\left[\lambda_m-\frac{q}{T},\lambda_m+\frac{q}{T}\right]} G_j\,
\end{equation}
under the assumption $E_j<\frac{p_{\min}-\ptail}{8}$. The rest of the proof is the same as in \Cref{thm:regime_1}. We also consider two classes of candidates:
\begin{itemize}
    \item When $\theta_j\in \cup_{m\in\mathcal{D}}\left[\lambda_m-\frac{q}{T},\lambda_m+\frac{q}{T}\right]\mod 2\pi$, we obtain
    \[
    \begin{aligned}
    G_j&=G(\theta_j)\geq \left|\sum^M_{m=1}p_m\phi_p\left(\lambda_m-\theta_j\right)\right|-E_j\\
    &> p_{\min}\exp\left(-\frac{q^2}{2}\right)-\frac{p_{\min}-\ptail}{8}\\
    &\geq  \ptail+\frac{p_{\min}-\ptail}{2}-\frac{p_{\min}-\ptail}{8}\geq \ptail+\frac{3(p_{\min}-\ptail)}{8}\,.
    \end{aligned}
    \]
    where we used \eqref{ineq:gauss bound} and $E_j<\frac{p_{\min}-\ptail}{8}$ in the second inequality, and the condition of $q$ in the last inequality.

    \item When $\theta_j\notin \cup_{m\in\mathcal{D}}\left[\lambda_m-\frac{\alpha}{3T},\lambda_m+\frac{\alpha}{3T}\right]\mod 2\pi$, we obtain
    \[
    \begin{aligned}
    G_j&\leq \left|\sum_{m\in\mathcal{D}}p_m\phi_p\left(\lambda_m-\theta_j\right)\right|+\left|\sum_{m\in\mathcal{D}^ c}p_m\phi_p\left(\lambda_m-\theta_j\right)\right|+E_j\\
    &\leq 1.01\exp\left(-\frac{\alpha^2}{18}\right)+\ptail+\frac{p_{\min}-\ptail}{8}\leq \ptail+\frac{p_{\min}-\ptail}{4}\,,
    \end{aligned}
    \]
    where we used \Cref{lemma:periodic guassian bound} and $E_j<\frac{p_{\min}-\ptail}{8}$ in the second inequality and the condition of $\alpha$ in the last inequality.
\end{itemize}
Therefore, \eqref{eqn:int_separation_regime_1} is proved, and we complete the proof of the theorem.

\end{proof}

\end{document}